\documentclass{article}

\usepackage{microtype}
\usepackage{graphicx}
\usepackage{subfigure}
\usepackage{booktabs} % for professional tables

\usepackage{amsmath, amsthm, amssymb}

\newcommand{\D}{\mathcal{D}}
\newcommand{\R}{\mathbb{R}}
\renewcommand{\S}{\mathcal{S}}

\newcommand{\x}{\mathbf{x}}
\newcommand{\y}{\mathbf{y}}
\newcommand{\z}{\mathbf{z}}
\newtheorem{theorem}{Theorem}
\newtheorem{lemma}{Lemma}

\newtheorem{cor}{Corollary}
\usepackage{enumitem}

\usepackage{hyperref}

\usepackage[accepted]{icml2020}

\icmltitlerunning{Graph-based Nearest Neighbor Search: From Practice to Theory}

\usepackage{color}

\begin{document}

\twocolumn[
\icmltitle{Graph-based Nearest Neighbor Search: \\ From Practice to Theory}

% It is OKAY to include author information, even for blind
% submissions: the style file will automatically remove it for you
% unless you've provided the [accepted] option to the icml2020
% package.

% List of affiliations: The first argument should be a (short)
% identifier you will use later to specify author affiliations
% Academic affiliations should list Department, University, City, Region, Country
% Industry affiliations should list Company, City, Region, Country

% You can specify symbols, otherwise they are numbered in order.
% Ideally, you should not use this facility. Affiliations will be numbered
% in order of appearance and this is the preferred way.
\icmlsetsymbol{equal}{*}

\begin{icmlauthorlist}
\icmlauthor{Liudmila Prokhorenkova}{equal,ya,mipt,hse}
\icmlauthor{Aleksandr Shekhovtsov}{equal,ya,hse}
\end{icmlauthorlist}

\icmlaffiliation{ya}{Yandex, Moscow, Russia}
\icmlaffiliation{mipt}{Moscow Institute of Physics and Technology, Dolgoprudny, Moscow Region, Russia}
\icmlaffiliation{hse}{Higher School of Economics, Moscow, Russia}

\icmlcorrespondingauthor{Liudmila Prokhorenkova}{ostroumova-la@yandex-team.ru}
%\icmlcorrespondingauthor{Eee Pppp}{ep@eden.co.uk}

% You may provide any keywords that you
% find helpful for describing your paper; these are used to populate
% the "keywords" metadata in the PDF but will not be shown in the document
\icmlkeywords{graph-based nearest-neighbor search, theory}

\vskip 0.3in
]

% this must go after the closing bracket ] following \twocolumn[ ...

% This command actually creates the footnote in the first column
% listing the affiliations and the copyright notice.
% The command takes one argument, which is text to display at the start of the footnote.
% The \icmlEqualContribution command is standard text for equal contribution.
% Remove it (just {}) if you do not need this facility.

%\printAffiliationsAndNotice{}  % leave blank if no need to mention equal contribution
\printAffiliationsAndNotice{\icmlEqualContribution} % otherwise use the standard text.

\begin{abstract}
Graph-based approaches are empirically shown to be very successful for the nearest neighbor search (NNS). However, there has been very little research on their theoretical guarantees. We fill this gap and rigorously analyze the performance of graph-based NNS algorithms, specifically focusing on the low-dimensional ($d \ll \log n$) regime. In addition to the basic greedy algorithm on nearest neighbor graphs, we also analyze the most successful heuristics commonly used in practice: speeding up via adding shortcut edges and improving accuracy via maintaining a dynamic list of candidates. We believe that our theoretical insights supported by experimental analysis are an important step towards understanding the limits and benefits of graph-based NNS algorithms.
\end{abstract}

\section{Introduction}\label{sec:intro}

Many methods in machine learning, pattern recognition, coding theory, and other research areas are based on nearest neighbor search~\cite{bishop2006pattern,chen2018explaining,may2015computing,shakhnarovich2006nearest}. In particular, the $k$-nearest neighbor method is included in the list of top 10 algorithms in data mining~\cite{wu2008top}. Since modern datasets are mostly massive (both in terms of the number of elements $n$ and the dimension $d$), reducing the computation complexity of NNS algorithms is of the essence. The nearest neighbor problem is to preprocess a given dataset $\D$ in such a way that for an arbitrary forthcoming query vector $q$ we can quickly (in time $o(n)$) find its nearest neighbors in $\D$. 

Many efficient methods exist for NN problem, especially when the dimension $d$ is small~\cite{arya1998optimal,bentley1975multidimensional,meiser1993point}. In particular, the algorithms based on recursive partitions of the space, like k-d threes and random projection trees, are widely used~\cite{bentley1975multidimensional,dasgupta2008random,dasgupta2015randomized,keivani2018improved}. The most well-known algorithm usable for large $d$ is the Locality Sensitive Hashing (LSH)~\cite{indyk1998approximate}, which is well studied theoretically and widely used in practical applications.
	
Recently, graph-based approaches were shown to demonstrate superior performance over other types of algorithms in many large-scale applications of NNS~\cite{aumuller2019ann}. Most graph-based methods are based on constructing a nearest neighbor graph (or its approximation), where nodes correspond to the elements of $\D$, and each node is connected to its nearest neighbors by directed edges~\cite{dong2011efficient,hajebi2011fast,wang2012scalable}.
Then, for a given query $q$, one first takes an element in $\D$ (either random or fixed predefined) and makes greedy steps towards $q$ on the graph: at each step, all neighbors of a current node are evaluated, and the one closest to $q$ is chosen. 
Various additional heuristics were proposed to speed up graph-based search~\cite{baranchuk2019learning,fu2019fast,iwasaki2018optimization,malkov2014approximate,malkov2018efficient}.

While there is a lot of evidence empirically showing the superiority of graph-based NNS algorithms in practical applications, there is very little theoretical research supporting this.
\citet{laarhoven2018graph} made the first step in this direction by providing time--space trade-offs for approximate nearest neighbor (ANN) search on sparse datasets uniformly distributed on a $d$-dimensional Euclidean sphere.
The current work significantly extends these results and differs in several important aspects, as discussed in the next section.

Our analysis assumes the uniform distribution on a sphere, and we mostly focus on the dense regime $d \ll \log n$.
This setup is motivated by our experiments demonstrating that \textit{uniformization} and \textit{densification} applied to a \textit{general} dataset may improve some graph-based NNS algorithms. 
The dense regime is additionally motivated by the fact that real-world datasets are known to have low intrinsic dimension~\cite{beygelzimer2006cover,lin2019graph}.
%,karger2002finding}

We first investigate how the greedy search on simple NN graphs work in dense and sparse regimes (Section~\ref{sec:plain}). For the dense regime, the time complexity is $\Theta\left(n^{1/d} \cdot M^d\right)$ for some constant $M$. Here $M^d$ corresponds to the complexity of one step and $n^{1/d}$ to the number of steps.

Then, we analyze the effect of long-range links (Section~\ref{sec:long-links}): in practice, shortcut edges between distant elements are added to NN graphs to make faster progress in the early stages of the algorithm. This is a core idea of, e.g., navigable small-world graphs (NSW and HNSW)~\cite{malkov2014approximate,malkov2018efficient}. 
While a naive method of adding long links uniformly at random does not improve the asymptotic query time, we show that properly added edges can reduce the number of steps to $O\left(\log^2 n\right)$. This part is motivated by~\citet{kleinberg2000small}, who proved a related result for a greedy routing on a two-dimensional lattice.
We adapt Kleinberg's result to NNS on a dataset uniformly distributed over a sphere in $\R^d$ (with possibly growing $d$). Additionally, while Kleinberg's theorem inspired many algorithms~\cite{beaumont2007peer,karbasi2015small,malkov2018efficient}, 
to the best of our knowledge,
we are the first to propose an efficient method of constructing properly distributed long edges for a \textit{general} dataset.

We also consider a heuristics known to significantly improve the accuracy: maintaining a dynamic list of several candidates instead of just one optimal point, which is a general technique known as \textit{beam search}. 
We rigorously analyze the effect of this technique, see Section~\ref{sec:beam-search}.

Finally, in Section~\ref{sec:experiments}, we empirically illustrate the obtained results, discuss the most promising techniques, and demonstrate why our assumptions are reasonable.

\section{Related Work}

Several well-known algorithms proposed for NNS are based on recursive partitions of the space, e.g., k-d trees and random projection trees~\cite{bentley1975multidimensional,dasgupta2008random,dasgupta2015randomized,keivani2018improved}. The query time of a k-d tree is $O\left(d\cdot 2^{O(d)}\right)$, which leads to an efficient algorithm for the exact NNS in the dense regime $d \ll \log n$~\cite{chen2018explaining}.
When $d \gg \log n$, all algorithms suffer from the curse of dimensionality, so the problem is often relaxed to approximate nearest neighbor (ANN) formally defined in Section~\ref{sec:setup}.
Among the algorithms proposed for the ANN problem, LSH~\cite{indyk1998approximate} is the most theoretically studied one. The main idea of LSH is to hash the points such that the probability of collision is much higher for points that are close to each other than for those that are far apart. Then, one can retrieve near neighbors for a query by hashing it and checking the elements in the same buckets. LSH solves $c$-ANN with query time $\Theta\left(d n^{\varphi}\right)$ and space complexity $\Theta\left(n^{1+\varphi}\right)$. Assuming the Euclidean distance, the optimal $\varphi$ is about $\frac{1}{c^2}$ for data-agnostic algorithms~\cite{andoni2008near,datar2004locality,motwani2007lower,o2014optimal}. By using a data-dependent construction, this bound can be improved up to $\frac{1}{2c^2 - 1}$~\cite{andoni2015optimal,andoni2015tight}. \citet{becker2016new} analyzed spherical locally sensitive filters (LSF) and showed that LSF can outperform LSH when the dimension $d$ is logarithmic in the number of elements $n$. LSF can be thought of as lying in the middle between LSH and graph-based methods: LSF uses small spherical caps to define filters, while graph-based methods also allow, roughly speaking, moving to the neighboring caps. 
	
In contrast to LSH, graph-based approaches are not so well understood.
It is known that if we want to be able to find the exact nearest neighbor for any possible query via a greedy search, then the graph has to contain the classical Delaunay graph as a subgraph.
However, for large $d$, Delaunay graphs have huge node degrees and cannot be constructed in a reasonable time~\cite{navarro2002searching}.
In a more recent theoretical paper, \citet{laarhoven2018graph} considers datasets uniformly distributed on a $d$-dimensional sphere with $d \gg \log n$ and provides time--space trade-offs for the ANN search. The current paper, while using a similar setting, differs in several important aspects. First, we mostly focus on the dense regime $d \ll \log n$, show that it significantly differs both in terms of techniques and results, and argue why it is reasonable to work in this setting. Second, in addition to the analysis of plain NN graphs, we are the first to analyze how some additional tricks commonly used in practice affect the accuracy and complexity of the algorithm. These tricks are adding shortcut edges and using beam search. Finally, we support some claims made in the previous work by a rigorous analysis (see 
%Appendix~\ref{sec:comp_w_laarh} 
Supplementary Materials~E).
Let us also briefly discuss a recent paper~\cite{fu2019fast}, which rigorously analyzes the time and space complexity for searching the elements of so-called monotonic graphs. Note that the obtained guarantees are provided only for the case when a query coincides with an element of a dataset, which is a very strong assumption, and it is not clear how the results would generalize to a general setting.

A part of our research is motivated by~\citet{kleinberg2000small}, who proved that after adding random edges with a particular distribution to a lattice, the number of steps needed to reach a query via a greedy routing scales polylogarithmically, see details in Section~\ref{sec:long-links}. 
We generalize this result to a more realistic setting and propose a practically applicable method to generate such edges.

Finally, let us mention a recent empirical paper comparing the performance of HNSW with other graph-based NNS algorithms~\cite{lin2019graph}. In particular, it shows that HNSW has superior performance over one-layer graphs only for low-dimensional data. We demonstrate theoretically why this can be the case (assuming uniform datasets): we prove that for plain NN graphs and for $d \gg \sqrt{\log n}$, the number of steps of graph-based NNS is negligible compared with one-step complexity. Moreover, for $d \gg \log n$, the algorithm converges in just two steps.
Interestingly, Figure~5 in~\citet{lin2019graph} shows that on uniformly distributed synthetic datasets simple kNN graphs and HNSW show quite similar results, especially when the dimension is not too small. This means that studying simple NN graphs is a reasonable first step in the analysis of graph-based NNS algorithms.
In contrast, on real datasets, HNSW is significantly better, which is caused by a so-called diversification of neighbors. However, as we show empirically in Section~\ref{sec:experiments}, simple kNN graphs can work comparably to HNSW even on real datasets, if we use the uniformization procedure~\cite{sablayrolles2018spreading}.

\section{Setup and Notation}\label{sec:setup}
	
We are given a dataset $\D = \{\x_1, \ldots \x_n\}$, $\x_i \in \R^{d+1}$ and assume that all elements of $\D$ belong to a unit Euclidean sphere, $\D \subset \S^{d}$. This special case is of particular importance for practical applications since feature vectors are often normalized.\footnote{From a theoretical point of view, there is a reduction of ANN in the entire Euclidean space to ANN on a sphere~\cite{andoni2015optimal,valiant2015finding}.} For a given query $q \in \S^d$ let $\bar \x \in \D$ be its nearest neighbor. The aim of the exact NNS is to return $\bar\x$, while in $c,R$-ANN (approximate \textit{near} neighbor), for given $R > 0$, $c>1$, we need to find such $\x'$ that $\rho(q,\x') \le cR$ if $\rho(q,\bar \x) \le R$~\cite{andoni2017optimal,chen2018explaining}.\footnote{There is also a notion of $c$-ANN, where the goal is to find such $\x'$ that $\rho(q,\x') \le c \rho (q,\bar\x)$; $c$-ANN can be reduced to $c,R$-ANN with additional $O(\log n)$ factor in query time and $O(\log^2 n)$ in storage cost~\cite{chen2018explaining,har2012approximate}.}  By $\rho(\cdot, \cdot)$ we further denote a spherical distance.
	
Similarly to~\citet{laarhoven2018graph}, we assume that the elements $\x_i \in \D$ are i.i.d.~random vectors uniformly distributed on $\S^{d}$. Random uniform datasets are considered to be the most natural ``hard'' distribution for the ANN problem~\cite{andoni2015optimal}. Hence, it is an important step towards understanding the limits and benefits of graph-based NNS algorithms.\footnote{Also,~\citet{andoni2015optimal} show how to reduce ANN on a generic dataset to ANN on a ``pseudo-random'' dataset for a data-dependent LSH algorithm.}
From a practical point of view, real datasets are usually far from being uniformly distributed. However, in some applications uniformity is helpful, and there are techniques allowing to make a dataset more uniform while approximately preserving the distances~\cite{sablayrolles2018spreading}. Remarkably, in our experiments (Section~\ref{sec:experiments}) we show that this trick, combined with dimensionality reduction, is beneficial for some graph-based NNS algorithms.
We further assume that a query vector $q \in \S^d$ is placed uniformly within a distance $R$ from the nearest neighbor $\bar \x$ (since $c,R$-ANN problem is defined conditionally on the event $\rho(q,\bar\x) \le R$). Such nearest neighbor is called \textit{planted}. 
	
In the current paper, we use a standard assumption that the dimensionality $d = d(n)$ grows with $n$~\cite{chen2018explaining}. We distinguish three fundamentally different regimes in NN problems: dense with $d \ll \log(n)$; sparse with $d \gg \log(n)$; moderate with $d = \Theta(\log(n))$.\footnote{Throughout the paper $\log$ is a binary logarithm (with base 2) and $\ln$ is a natural logarithm (with base $e$).} 
While~\citet{laarhoven2018graph} focused solely on the sparse regime, we also consider the dense one, which is fundamentally different in terms of geometry, proof techniques, and query time.

As discussed in the introduction, most graph-based approaches are based on constructing a nearest neighbor graph (or its approximation). For uniformly distributed datasets, connecting an element $\x$ to a given number of nearest neighbors is essentially equivalent to connecting it to all such nodes $\y$ that $\rho(\x,\y) \le \rho^*$ with some appropriate $\rho^*$ (since the number of nodes at a distance at most $\rho^*$ is concentrated around its expectation, see also Supplementary Materials~F.3 for an empirical illustration). Therefore, at the preprocessing stage, we choose some $\rho^*$ and construct a graph using this threshold. Later, when we get a query $q$, we sample a random element $\x \in \D$ such that $\rho(\x,q) < \frac{\pi}{2}$,\footnote{We can easily achieve this by trying a constant number of random samples since each sample succeeds with probability $1/2$. This trick speeds up the procedure (without affecting the asymptotics) and simplifies the proofs.}
and perform a graph-based greedy descent: at each step, we measure the distances between the neighbors of a current node and $q$ and move to the closest neighbor, while we make progress. 
In the current paper, we assume one iteration of this procedure, i.e., we do not restart the process several times. 
	
\section{Theoretical Analysis}\label{sec:analysis}
	
In this section, we overview the obtained theoretical results. 
We start with a preliminary analysis of the properties of spherical caps and their intersections. These properties give some intuition and are extensively used throughout the proofs.
Then, we analyze the performance of the greedy search over plain NN graphs. We mostly focus on the dense regime, but also formulate the corresponding theorem for the sparse one. After that, we analyze how the shortcut edges affect the complexity of the algorithm. We prove that they indeed improve the asymptotic query time, but only in the dense regime. Finally, we analyze the positive effect of the beam search, which is widely used in practice. 

\subsection{Auxiliary Results}

Let us discuss some technical results on the volumes of spherical caps and their intersections, which we extensively use in the proofs.

Denote by $\mu$ the Lebesgue measure over $\R^{d+1}$. By $C_\x(\gamma)$ we denote a spherical cap of \textit{height} $\gamma$ centered at $\x \in \S^d$, i.e., $\{\y \in \S^d: \langle \x, \y \rangle \ge \gamma\}$; $C(\gamma) = \mu\left(C_\x(\gamma)\right)$ denotes the volume of a spherical cap of height $\gamma$. Throughout the paper, for any variable $\gamma$, $0 \le \gamma \le 1$, we let $\hat \gamma := \sqrt{1 - \gamma^2}$. We say that $\hat \gamma$ is the \textit{radius} of a spherical cap.

The volume of a spherical cap defines the expected number of neighbors of a node. We estimate this volume in Supplementary Materials~A.1.  Despite some additional terms (which can often be neglected), one can think that $C(\gamma) \propto \hat \gamma^{d}$.

The intersections of spherical caps are also important: their volumes are needed to estimate the probably of making a step of graph-based search. Formally, by $W_{\x,\y}(\alpha,\beta)$ we denote the intersection of two spherical caps centered at $\x \in \S^d$ and $\y \in \S^d$ with heights $\alpha$ and $\beta$, respectively, i.e., $W_{\x,\y}(\alpha,\beta) = \{\z \in \S^d: \langle \z,\x \rangle \ge \alpha, \langle \z,\y \rangle \ge \beta \}$. By $W(\alpha,\beta,\theta)$ we denote the volume of such intersection given that the angle between $\x$ and $\y$ is $\theta$. In Supplementary Materials~A.2, we estimate $W(\alpha,\beta,\theta)$. We essentially prove that for $\gamma = \frac{\sqrt{\alpha^2 + \beta^2 - 2 \alpha\beta\cos\theta}}{\sin\theta}$, $W(\alpha, \beta, \theta)$ (or its complement) $\propto \hat\gamma^{d}$.

Although our results are similar to those formulated by~\citet{becker2016new}, it is crucial for our analysis that parameters defining spherical caps depend on $d$ and either $\gamma$ or $\hat \gamma$ may tend to zero. In contrast, the results of~\citet{becker2016new} hold only for fixed parameters. 
Also, we extend Lemma 2.2 in their paper by analyzing both $W(\alpha, \beta, \theta)$ and its complement: we need a lower bound on $W(\alpha, \beta, \theta)$ to show that with high probability we can make a step of the algorithm and an upper bound on $C(\alpha) - W(\alpha, \beta, \theta)$ to show that at the final step of the algorithm we can find the nearest neighbor with high probability.

The fact that $C(\gamma) \propto \hat{\gamma}^d$ allows us to understand the principal difference between the dense and sparse regimes.
For dense datasets, we assume $d = d(n) = \log(n)/\omega$, where $\omega = \omega(n) \to \infty$ as $n \to \infty$.
In this case, it is convenient to operate with radii of spherical caps. 
An essential property of the dense regime is the fact that the distance from a given point to its nearest neighbor behaves as $2^{-\omega}$, so it decreases with $n$. Indeed, let $\hat\alpha_1$ be the radius of a cap centered at a given point and covering its nearest neighbor, then we have $C(\alpha_1) \sim \frac 1 n$, i.e., $\hat\alpha_1 \sim n^{-\frac{1}{d}} = 2^{-\omega}$. To construct a nearest neighbor graph, we use spherical caps of radius $M \cdot 2^{-\omega}$ with some constant $M>1$.

In sparse regime, we assume $d = \omega\log(n)$, $\omega \to \infty$ as $n \to \infty$.
In this case, it is convenient to operate with heights of spherical caps. A crucial feature of sparse regime is that the heights under consideration tend to zero as $n$ grows. Informally speaking, all nodes are at spherical distance about $\pi/2$ from each other. Indeed, we have $C(\alpha_1) \sim \frac 1 n$, i.e., $\alpha_1^2 \sim 1 - n^{-\frac{2}{d}} = 1 - 2^{-\frac{2}{\omega}} \propto w^{-1}$
and it tends to zero with $n$.
%\sim \frac{2}{w \log e} 
%\to 0$ as $n \to \infty$. 
In this case, to construct a graph, we use spherical caps with height equal to
%$M \cdot \frac{2}{w \log e}$, 
$\sqrt{M} \cdot \alpha_1$,
where $M<1$ is some constant.  

\subsection{Greedy Search on Plain NN Graphs}\label{sec:plain}

\paragraph{Dense regime}

For any constant $M > 1$, let $G(M)$ be a graph obtained by connecting $\x_i$ and $\x_j$ iff $\rho(\x_i,\x_j) \le \arcsin\left( M \, n^{-1/d}\right)$. The following theorem is proven in Supplementary Materials~B.1-B.3.

\begin{theorem}\label{thm:dense_main}
Assume that $d \gg \log \log n$ and we are given some constant $c \ge 1$. Let $M$ be a constant such that $M > \sqrt{\frac{4c^2}{3c^2 - 1}}$, then, with probability $1-o(1)$, $G(M)$-based NNS solves $c,R$-ANN for any $R$ (or the exact NN problem if $c=1$); time complexity is $\Theta\left(d^{1/2}\cdot n^{1/d} \cdot M^d \right) = n^{o(1)}$; space complexity is $\Theta\left(n\cdot d^{-1/2}\cdot M^d\cdot \log n\right) =  n^{1+o(1)}$. 
\end{theorem}

It follows from Theorem~\ref{thm:dense_main} that both time and space complexities increase with $M$ (for constant $M$), so one may want to choose the smallest possible $M$. When the aim is to find the exact nearest neighbor ($c = 1$), we can take any $M>\sqrt{2}$. When $c > 1$, the lower bound for $M$ decreases with $c$.
	
The space complexity straightforwardly depends on $M$: the radius of a spherical cap defines the expected number of neighbors for a node, which is $\Theta(d^{-1/2}\cdot M^d)$, and additional $\log n$ corresponds to storing integers up to $n$. The time complexity consists of two terms: the complexity of one step is multiplied by the number of steps. The complexity of one step is the number of neighbors multiplied by $d$: $\Theta(d^{1/2}\cdot M^d)$.
The number of steps is $\Theta\left(n^{1/d}\right)$. 

While $d$ is negligible compared with $M^d$, the relation between $n^{1/d}$ and $M^d$ is non-trivial. Indeed, when $d \gg \sqrt{\log n}$, the term $M^d$ dominates, so in this regime, the smaller $M$ the better asymptotics we get (both for time and space complexities). However, when the dataset is very dense, i.e., $d \ll \sqrt{\log n}$, the number of steps becomes much larger than the complexity of one step. 
For such datasets, it could be possible that taking $M = M(n) \gg 1$ would improve the query time (as it affects the number of steps). However, in Supplementary Materials~B.4,
%\ref{sec:larger_neighborhood}, 
we prove that this is not the case, and the query time from Theorem~\ref{thm:dense_main} cannot be improved.

Finally, since all distances considered in the proof tend to zero with $n$, it is easy to verify that all results stated above for the spherical distance also hold for the  Euclidean one.

\paragraph{Sparse regime}

For any $M$, $0 < M < 1$, let $G(M)$ be a graph obtained by connecting $\x_i$ and $\x_j$ iff $\rho(\x_i,\x_j) \le \arccos\left(\sqrt{\frac{2M\ln n}{d}}\right)$. The following theorem holds (see Supplementary Materials~B.5 for the proof).

\begin{theorem}\label{thm:sparse_main}
For any $c > 1$ let $\alpha_c = \cos\left(\frac{\pi}{2c}\right)$ and  let $M$ be any constant such that $M < \frac{\alpha_c^2}{\alpha_c^2 + 1}$. Then, with probability $1-o(1)$, $G(M)$-based NNS solves $c,R$-ANN (for any $R$ and for spherical distance); time complexity of the procedure is $\Theta\left(n^{1-M+o(1)} \right)$; space complexity is $\Theta\left(n^{2-M+o(1)}\right)$.
\end{theorem}

Interestingly, as follows from the proof, in the sparse regime, the greedy algorithm converges in at most two steps with probability $1-o(1)$ (on a uniform dataset). As a result, there is no trade-off between time and space complexity: larger values of $M$ reduce both of them. 

One can easily obtain an analog of Theorem~\ref{thm:sparse_main} for the Euclidean distance. In Theorem~\ref{thm:sparse_main}, $\alpha_c$ is the height of a spherical cap covering a spherical distance $\frac{\pi}{2c}$, which is $c$ times smaller than $\pi/2$. For the Euclidean distance, we have to replace $\frac{\pi}{2}$ by $\sqrt{2}$ and then the height of a spherical cap covering Euclidean distance $\sqrt{2}/c$ is $\alpha_c = 1 - \frac{1}{c^2}$. So, we get the following corollary. 
\begin{cor}
For the Euclidean distance, Theorem~\ref{thm:sparse_main} holds with $\alpha_c = 1 - \frac{1}{c^2}$, i.e., $M < \frac{\left(1 - 1/c^2\right)^2}{\left(1 - 1/c^2\right)^2 + 1}$.
\end{cor}    
As a result, we can obtain time complexity $n^{\varphi}$ and space complexity $n^{1 + \varphi}$, where $\varphi$ can be made about $\frac{1}{\left(1 - 1/c^2\right)^2 + 1}$.
Note that this result corresponds to the case $\rho_q = \rho_s$ in~\citet{laarhoven2018graph}.

\subsection{Analysis of Long-range Links}\label{sec:long-links}

As discussed in the previous section, when $d \gg \sqrt{\log n}$, the number of steps is negligible compared to the one-step complexity. Hence, reducing the number of steps cannot change the main term of the asymptotics.\footnote{This agrees with the empirical results obtained by~\citet{lin2019graph}, where on synthetic datasets the difference between simple kNN graphs and the more advanced HNSW algorithm becomes smaller as $d$ increases and vanishes after $d = 8$.}
However, if $d \ll \sqrt{\log n}$ (very dense setting), the number of steps becomes the main term of time complexity. 
In this case, it is reasonable to reduce the number of steps via adding so-called long-range links (or shortcuts)~--- some edges connecting elements that are far away from each other~--- which may speed up the search on early stages of the algorithm.

The simplest way to obtain a graph with a small diameter from a given graph is to connect each node to a few neighbors chosen uniformly at random. This idea is proposed by~\citet{watts1998collective} and gives $O\left(\log n \right)$ diameter for the so-called ``small-world model''.  However, having a logarithmic diameter does \textit{not} guarantee a logarithmic number of steps in graph-based NNS, since these steps, while being greedy in the underlying metric space, may not be optimal on a graph~\cite{kleinberg2000small}. In Supplementary Materials~C.1, we formally prove that adding edges uniformly at random cannot improve the asymptotic time complexity. The intuition is simple: choosing the closest neighbor among the long-range ones is equivalent to sampling a certain number of nodes (uniformly at random among the whole set) and then choosing the one closest to~$q$.

This agrees with~\citet{kleinberg2000small}, who considered a 2-dimensional grid supplied with some random long-range edges. Kleinberg assumed that in addition to the local edges, each node creates one random outgoing long link, and the probability of a link from $u$ to $v$ is proportional to $\rho(u,v)^{-r}$. He proved that for $r = 2$, the greedy graph-based search finds the target element in $O\left(\log^2 n\right)$ steps, while any other $r$ gives at least $n^\varphi$ with $\varphi > 0$. This result can be easily extended to \textit{constant} $d>2$, in this case one should take $r = d$ to achieve polylogarithmic number of steps.

\citet{kleinberg2000small} influenced a large number of further studies. Some works generalized the result to other settings~\cite{barriere2001efficient,bonnet2007small,duchon2006could},  others used it as a motivation of search algorithms~\cite{beaumont2007peer,karbasi2015small}. It is also mentioned as a motivation for a widely used HNSW algorithm~\cite{malkov2018efficient}. However, Kleinberg's probabilities have not been used directly since 1) it is unclear how to use them for general distributions, when the intrinsic dimension is not known or varies over the dataset, 2) generating properly distributed edges has $\Theta\left(n^2d\right)$ complexity, which is infeasible for many practical tasks.

We address these issues. First, we translate the result of~\citet{kleinberg2000small} to our setting (importantly, we have $d \to \infty$). Second, we show how one can apply the method to general datasets without thinking about the intrinsic dimension and non-uniform distributions. Finally, we discuss how to reduce the computational complexity of graph construction procedure.

Following~\citet{kleinberg2000small}, we draw long-range edges with the following probabilities:
\begin{equation}\label{eq:kleinberg_prob}
\mathrm{P}(\text{edge from $u$ to $v$}) = \frac{\rho(u,v)^{-d}}{\sum_{w \neq u} \rho(u, w)^{-d}}\,.
\end{equation}

\begin{theorem}\label{thm:dense_kleinberg}
Under the conditions of Theorem~\ref{thm:dense_main}, sampling one long-range edge for each node according  to~\eqref{eq:kleinberg_prob} reduces the number of steps to $O( \log^2 n) $
(with probability $1 - o(1)$).
\end{theorem}

This theorem is proven in Supplementary Materials~C.2.
It is important that in contrast to~\citet{kleinberg2000small}, we assume $d \to \infty$. Indeed, a straightforward generalization of Kleinberg's result to non-constant $d$ gives an additional $2^d$ multiplier, which we were able to remove.

Note that using long-range edges, we can guarantee $O\left( \log^2 n\right)$ steps, while plain NN graphs give $\Theta\left(n^{1/d} \right)$ (see Theorem~\ref{thm:dense_main} and the discussion after that). Hence, reducing the number of steps is reasonable if $\log^2 n < n^{1/d}$, which means that $d < \frac{\log n}{2\log\log n}$.

As follows from the proof, adding more long-range edges may further reduce the number of steps. It is theoretically shown in the following corollary and is empirically evaluated in Supplementary Materials~F.3.
%Appendix~\ref{sec:experiments_app}.

\begin{cor} \label{cor:log_long}
If for each node we add $\Theta(\log n)$ long-range edges, then the number of steps becomes $O(\log n)$, so the time complexity is $O\left(d^{1/2}\cdot \log n \cdot M^d \right)$ while the asymptotic space complexity does not change compared to Theorem~\ref{thm:dense_main}. Further increasing the number of long-range edges does not improve the asymptotic complexity. 
\end{cor}

Also, we suggest the following trick, which noticeably reduces the number of steps in practice while not affecting the theoretical asymptotics: at each iteration, we check the long-range links first and proceed with the local ones only if we cannot make progress with long edges. We empirically evaluate this trick (called \textsc{llf}) in Section~\ref{sec:experiments}.

It is non-trivial how to apply Theorem~\ref{thm:dense_kleinberg} in practice due to the dependence of probabilities on $d$: real datasets usually have a low intrinsic dimension even when embedded to a higher-dimensional space~\cite{beygelzimer2006cover,lin2019graph}, and the intrinsic dimension may vary over the dataset. Let us show how to make the distribution in~\eqref{eq:kleinberg_prob} dimension-free. 
For lattices and uniformly distributed datasets, the distance to a node is strongly related to the \textit{rank} of this node in the list of all elements, when they ordered by a distance. Formally, let $v$ be the $k$-th neighbor of $u$. Then, we define $\rho_{rank}(u,v) = (k/n)^{1/d}$. For uniform datasets, $\rho(u,v)$ locally behaves as $\rho_{rank}(u,v)$ (up to a constant multiplier) since the number of nodes at a distance $\rho$ from a given one grows as $\rho^d$. If we replace $\rho$ by $\rho_{rank}$ in~\eqref{eq:kleinberg_prob}, we obtain: 
\begin{equation}\label{eq:kleinberg_rank_prob}
\mathrm{P}(\text{edge to $k$-th neighbor}) = \frac{1/k}{\sum_{i = 1}^n 1/i} \sim \frac{1}{k \ln n}\,.
\end{equation}
This distribution is \textit{dimension independent}, i.e., it can be easily used for general datasets.

Finally, we address the problem of the computational complexity of graph construction procedure. For this, we propose to generate long edges as follows: for each source node, we first choose $n^{\varphi}$, $0<\varphi<1$, candidates uniformly at random, and then sample a long edge from this set according to the probabilities in~\eqref{eq:kleinberg_rank_prob}. Pre-sampling reduces the time complexity to $\Theta\left(n^{1+\varphi} \, \left(d + \log n^{\varphi} \right)\right)$, compared with $\Theta(n^2 d)$ in Theorem~\ref{thm:dense_kleinberg}. The following lemma shows how it affects the distribution.

\begin{lemma}\label{lem:appr_kleinberg}
Let $P_k$ be the probability defined in~\eqref{eq:kleinberg_rank_prob} and $P_k^{\varphi}$ be the corresponding probability assuming the pre-sampling of $n^{\varphi}$ elements. Then, for all $k > n^{1-\varphi}$, $P_k^{\varphi}/P_k 
= \Theta(1/\varphi)$.
\end{lemma}

Informally, Lemma~\ref{lem:appr_kleinberg} says that for most of the elements, the probability does not change significantly.
In Section~\ref{sec:experiments}, we use pre-sampling with $\varphi = 1/2$ and in %Appendix~\ref{sec:experiments_app} 
Supplementary Materials~F.3
we show that this pre-sampling and using~\eqref{eq:kleinberg_rank_prob} instead of~\eqref{eq:kleinberg_prob} does not affect the quality of graph-based NNS.

\subsection{Beam Search}\label{sec:beam-search}

Beam search is a heuristic algorithm that explores a graph by expanding the most promising element in a limited set. It is widely used in graph-based NNS algorithms~\cite{fu2019fast,malkov2018efficient} as it drastically improves the accuracy. Moreover, it was shown that even a simple kNN graph supplied with beam search can show results close to the state-of-the-art on synthetic datasets~\cite{lin2019graph}. 
However, to the best of our knowledge, we are the first to analyze the effect of beam search in graph-based NNS algorithms theoretically.

Theorem~\ref{thm:dense_main} states that to solve the exact NN problem with probability $1-o(1)$, we need to have about $M^d$ neighbors for each node, where $M > \sqrt{2}$. If we reduce $M$ below this bound, then with a significant probability, the algorithm may get stuck in a local optimum before it reaches the nearest neighbor.
As follows from the proof, the problem of local optima is critical for the last steps of the algorithm, i.e., large degrees are needed near the query. 

Let us show that $M$ can be reduced if beam search is used instead of greedy search.
Assume that we have reached some local neighborhood of the query $q$ and there are $m-1$ points closer to $q$ in the dataset. 
Further, assume that we use beam search with at least $m$ best candidates stored. 
If the subgraph on $m$ nodes closest to $q$ is connected, then after at most $m$ steps the beam search terminates and returns $m$ closest nodes including the nearest neighbor (see Figure~\ref{fig:example} for an illustration).
Thus, we can reduce the size of the neighborhood, but instead of getting stuck in a local optimum, we explore the neighborhood until we reach the target.
Formally, the following theorem holds.

\begin{theorem}\label{thm:beam_search_common}
Let $M > 1$, $L > 1$ be such constants that $M^2 \left(1 - \frac{M^2}{4 L^2}\right) > 1$ and let $\log \log n \ll d \ll \log n $. Assume that we use beam search with $\frac{C\, L^d}{\sqrt{d}}$ candidates (for a sufficiently large $C$) and  we add $\Theta(\log n)$ long-range edges. Then, $G(M)$-based NNS solves the exact NN problem with probability $1-o(1)$. The time complexity is  $O\left( L^d \cdot M^d \right)$.
\end{theorem}

\begin{figure}
    \centering
    \includegraphics[width = \linewidth]{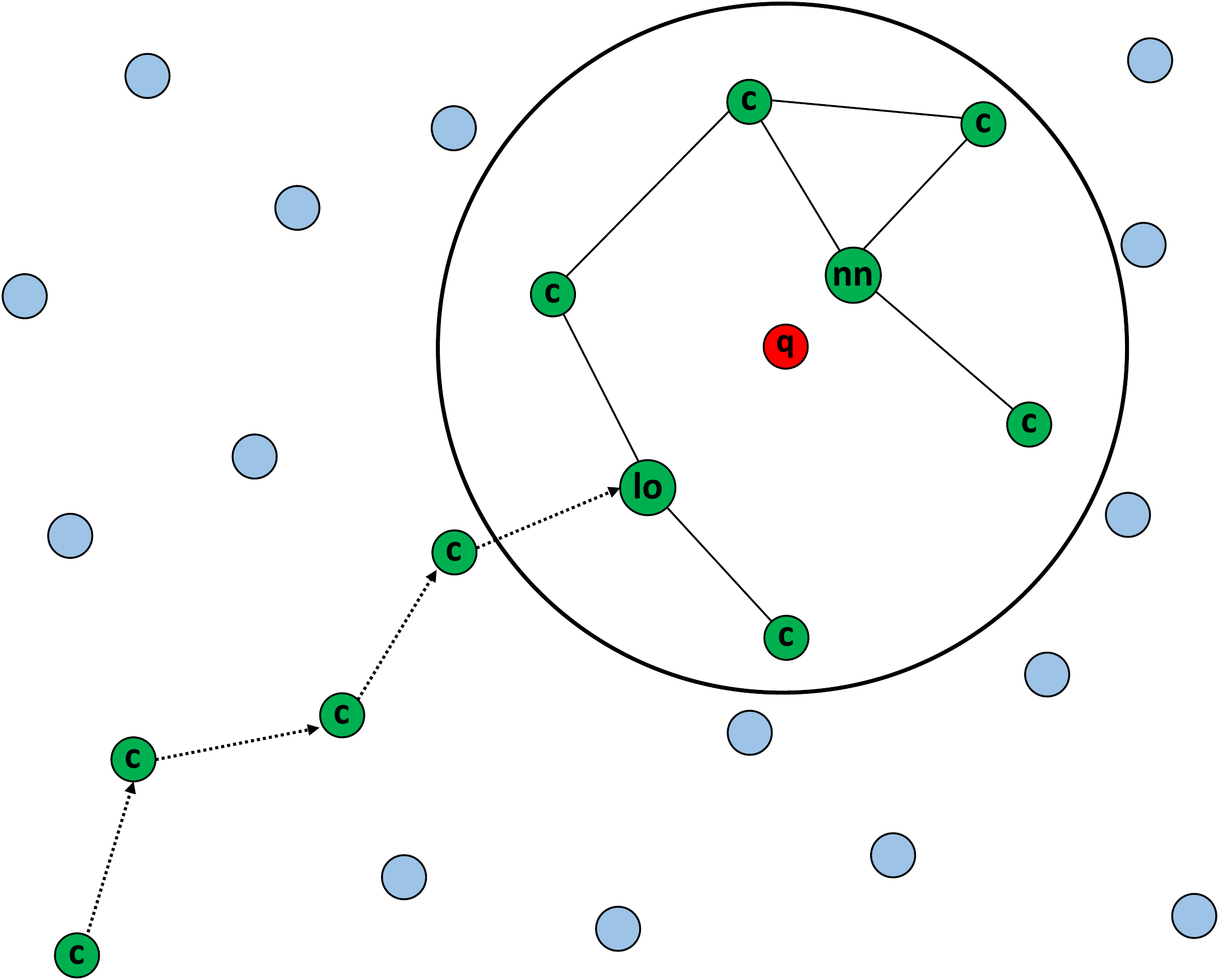}
    \caption{Example of beam search: ``q''~--- query, ``nn''~--- nearest neighbor, ``lo''~--- local optimum for greedy search, ``c''~--- other candidates visited during NNS. Beam search with 7 candidates returns 7 nearest elements to the query since the subgraph induced by them is connected.}
    \label{fig:example}
\end{figure}

As a result, beam search allows to significantly reduce degrees of a graph, which finally leads to time complexity reduction. 
We empirically illustrate this fact
in Table~\ref{tab:exp1} in the next section. 

Theorem~\ref{thm:beam_search_common} gives the time--space trade-off for the beam search: we can reduce $M$, which determines the space complexity, to any value grater than 1, but for small $M$ we have to take large $L$ which increases the query time. The following corollary optimizes the time complexity.

\begin{cor}
In Theorem~\ref{thm:beam_search_common}, we can take $M = \sqrt{\frac{3}{2}}$ and any $L > \sqrt{\frac{9}{8}}$. As a result, the main term of the time complexity can be reduced to $\left(\frac{27}{16}\right)^{d/2}$, which is less than $2^{d/2}$ from Theorem~\ref{thm:dense_main}.
\end{cor}

\section{Experiments}\label{sec:experiments}

In this section, we illustrate the obtained results using synthetic datasets and also analyze whether our findings translate to real data. 
We also demonstrate how uniformization and densification may improve the quality of some graph-based NNS algorithms, which motivates the assumptions used in our analysis. 

\subsection{Experimental Setup}

\paragraph{Datasets} To illustrate our theoretical results, we generated synthetic datasets uniformly distributed over $d$-dimensional spheres with $n = 10^6$ and $d \in \{2, 4, 8, 16 \}$. 
We also experimented with the widely used real-world datasets: SIFT and GIST~\cite{sift_gist}, GloVe~\cite{glove}, and DEEP1B~\cite{Babenko_2016_CVPR}. 
These datasets are summarized in Table~\ref{tab:datasets}.

\begin{table}
\caption{Real datasets}
\label{tab:datasets}
\vskip 0.15in
\centering
 \begin{tabular}{c c c c c} 
 \toprule
 dataset & dim  & \# base & \# queries & metric \\ 
 \midrule
 SIFT & 128  & $10^6$ & $10^4$ & Euclidean \\ 
 GIST & 960  & $10^6$ & $10^3$ & Euclidean\\
 GloVe & 300  & $10^6$ & $10^4$ & Angular \\
 DEEP & 96 & $10^6$ & $10^4$ & Euclidean \\
\bottomrule
\end{tabular}
\end{table}
%Supplementary Materials~F.1.
%Table~\ref{tab:datasets} in Appendices.

\begin{table}
\caption{\textsc{kNN} graphs on synthetic data}
\label{tab:exp1}
\vskip 0.15in
\vspace{4pt}
\centering
 \begin{tabular}{c c c c c} 
 \toprule
 dim & steps & degree & Recall@1 & Beam \\ %[0.5ex] 
 \midrule
 2 & 200 & 20 & 0.998 & 1 \\ 
 4 & 15 & 60 & 0.999 & 1 \\ 
 8 & 5 & 300 & 0.998 & 1 \\
 16 & 4 & 2000 & 0.99 & 1 \\
 16 & 106 & 20 & 0.995 & 100 \\ %[1ex] 
 \bottomrule
\end{tabular}
\end{table}

\paragraph{Measuring performance} To evaluate the algorithms, we adopt a standard approach and compare their Time-versus-Quality curves. 
On $x$-axis, instead of a frequently used $\mathrm{Recall}@1$, we have $\mathrm{error} = 1 - \mathrm{Recall}@1$ and use the logarithmic scale for better visualisation, since the most important region is where $\mathrm{error}$ is close to zero. On $y$-axis, for synthetic datasets, we show the number of distance calculations (the smaller the better), since we want to illustrate our theoretical results. For real datasets, we measure the number of queries per second (the larger the better). All algorithms were run on a single core of Intel(R) Core(TM) i7-6800K CPU @ 3.40GHz.
% Intel(R) Xeon(R) Gold 6230 CPU @ 2.10GHz. 

\paragraph{  Algorithms} 
In the experiments, we combine and analyze the techniques discussed in this paper:
\begin{itemize}[noitemsep]
    \item \textsc{kNN} is a simple NN graph, where each node has a fixed number of neighbors; 
    %\lt{ \textsc{$\rho$NN} is similar, but nodes are connected if they are closer than some threshold, it is used only on uniform datasets};
    \item \textsc{Kl} means that for each node we additionally add edges distributed according to~\eqref{eq:kleinberg_rank_prob}, we also use pre-sampling of $\sqrt{n}$ nodes to speed up the construction procedure;
    \item \textsc{llf}: at each iteration check \textit{long links first} and consider local neighbors only if required;
    \item \textsc{beam}: use beam search instead of greedy search, always used on real data;
    \item \textsc{dim-red}: dimensionality reduction technique used on real datasets, as discussed in more details in Section~\ref{sec:real_datasets}.
\end{itemize}

\begin{figure*}
    \centering
   \includegraphics[width=\textwidth]{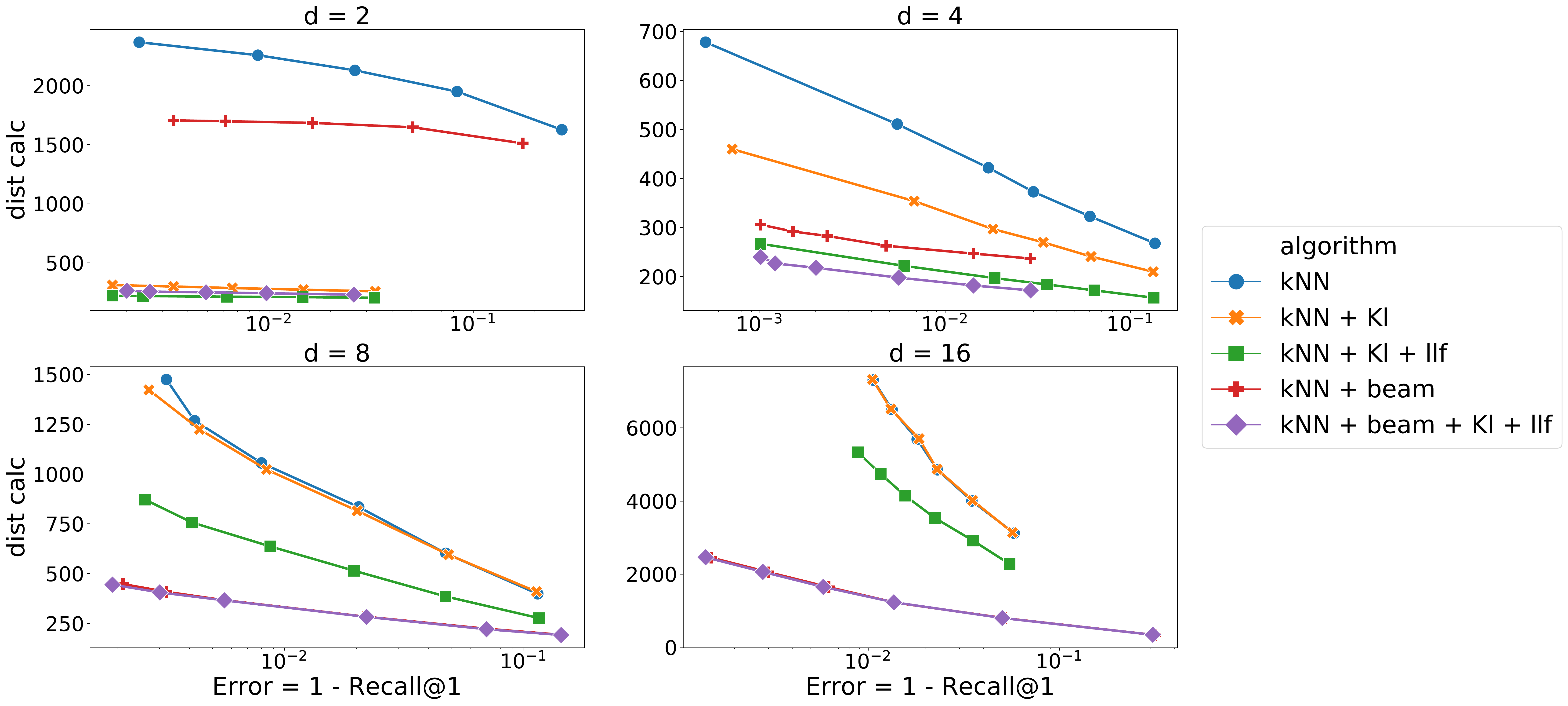}
    \caption{Results on synthetic datasets}
    \label{fig:synt}
\end{figure*}

To obtain Time-vs-Quality curves, we vary the number of local neighbors for greedy search and the number of candidates for beam search.

For reference, on real datasets, we compare all algorithms with HNSW~\cite{malkov2018efficient}, as it is shown to provide the state-of-the-art performance on the common benchmarks and its source code is publicly available.

\subsection{Synthetic Datasets}

In this section, we illustrate our theoretical results from Section~\ref{sec:analysis} on synthetic datasets. 
The source code of these experiments is publicly available.\footnote{\url{https://github.com/Shekhale/gbnns_theory}}

Figure~\ref{fig:synt} shows plain NN graphs (Theorems~\ref{thm:dense_main} and~\ref{thm:sparse_main}), the effect of adding long-range links (Theorem~\ref{thm:dense_kleinberg}), the \textsc{llf} heuristics, and the beam search (Theorem~\ref{thm:beam_search_common}).

In small dimensions, the largest improvement is caused by using properly distributed long edges. However, for larger $d$, adding such edges does not help much, which agrees with our theoretical results: when $d$ is large, the number of steps becomes negligible compared to the complexity of one step (see  Section~\ref{sec:long-links}).
Note that \textsc{llf} always improves the performance, which is expected, since \textsc{llf} reduces the number of distance computations for local neighbors.

In contrast, the effect of \textsc{beam} search is relatively small for $d = 2$, and it becomes substantial as $d$ grows. This agrees with our analysis in Section~\ref{sec:beam-search}: beam search helps to reduce the degrees, which are especially large for sparse datasets. 
Indeed, Theorems~\ref{thm:dense_main} and~\ref{thm:sparse_main} show that graph degrees are expected to explode with $d$.
Table~\ref{tab:exp1} additionally illustrates this:
we fix a sufficiently high value of Recall@1 and analyze which values of $k$ would allow \textsc{kNN} to reach such performance (approximately). We see that degrees explode: while 20 is sufficient for $d = 2$, we need 2000 neighbors for $d = 16$. In contrast, the number of steps decreases with $d$ from 200 to 4, which also agrees with our analysis. However, using the beam search with size 100 reduces the number of neighbors back to 20 while still having fewer steps compared to $d = 2$.

\begin{figure*}
    \centering
   \includegraphics[width=\textwidth]{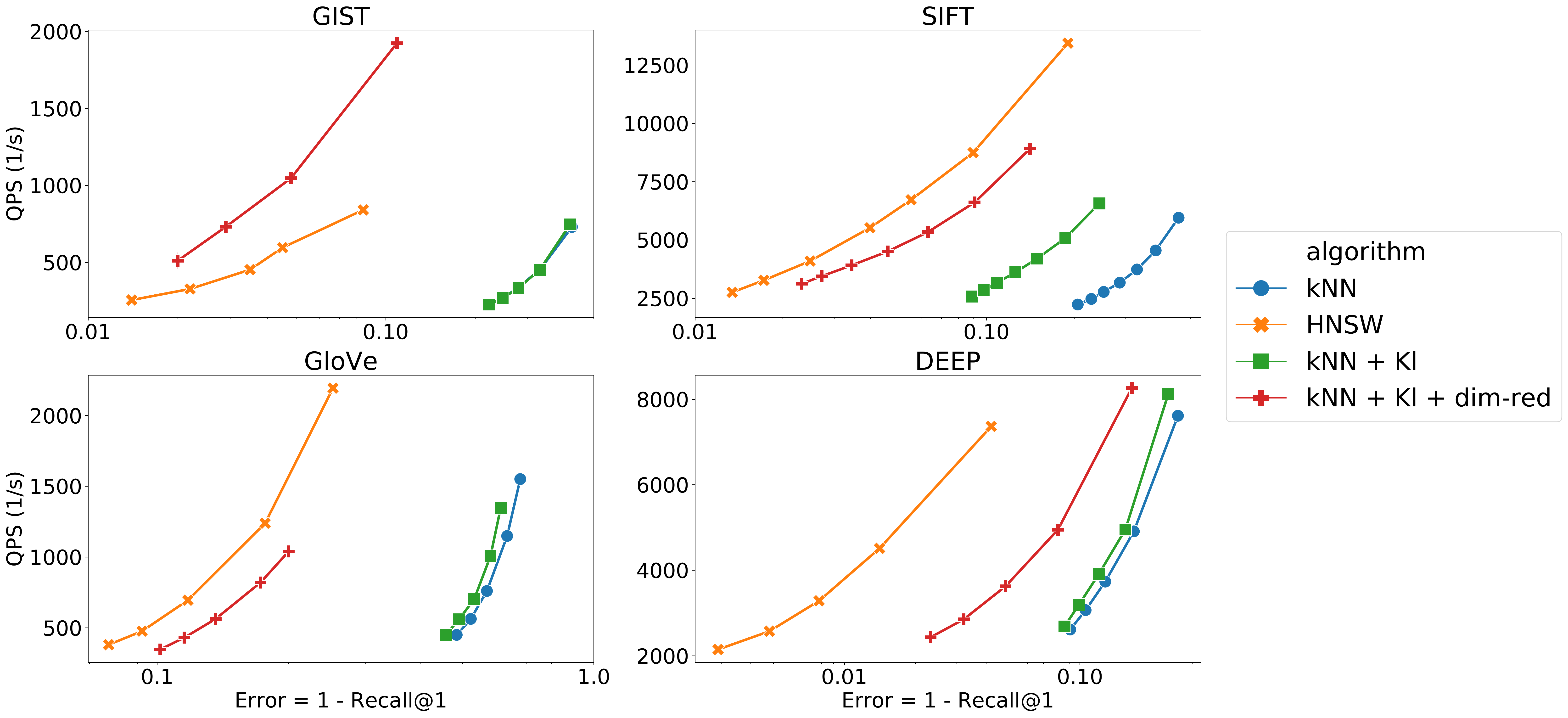}
    \caption{Results on real datasets}
    \label{fig:real}
\end{figure*}

\subsection{Real Datasets}\label{sec:real_datasets}

While our theoretical results assume uniformly distributed elements, in this section we analyze their viability in real settings. 

Let us first discuss a technique (\textsc{dim-red}) that exploits uniformization and densification and allows us to improve the quality of \textsc{kNN} significantly. 
This technique is inspired by our theoretical results and also by the observation that on uniform datasets for $d > 8$ plain NN graphs and HNSW have similar performance~\cite{lin2019graph}. The basic idea is to map a given dataset to a smaller dimension and at the same time make it more uniform while trying to preserve the neighborhoods. 
For this purpose, we use the technique proposed by~\citet{sablayrolles2018spreading}. 
At the preprocessing step, we construct a search graph on the new (smaller-dimensional) dataset. 
This dataset is dense and close to uniform, so our theoretical results are applicable. 
Then, for a given query $q$, we obtain a lower-dimensional $q'$ via the same transformation and then perform the beam search there. 
Finally, we come back to the original space and find the element closets to the original query from the beam candidates. This technique is called \textsc{dim-red}. 

We compare all algorithms in Figure~\ref{fig:real}. 
Recall that in all cases we use \textsc{beam} search since it drastically improves the greedy algorithm. 
We observe that long edges significantly improve the performance of \textsc{kNN} (in the original space) only for the SIFT dataset. This agrees with our analysis: SIFT has the smallest intrinsic dimension~\cite{lin2019graph}.
However, combining the algorithm with \textsc{dim-red}, we get a significant speedup. For GIST, the obtained algorithm is even better than HNSW.
This improvement may be caused by the fact that the original dimension is excessive, so we may reduce it without affecting the neighborhoods much. A smaller dimensionality is beneficial since it allows for faster distance computations and may also help to avoid local optima which are critical for sparse datasets.

These results confirm that analysis of uniform distributions and dense datasets is important: while these conditions are not satisfied in real datasets, they are close to being satisfied after the transformation. 

The source code for reproducing these experiments is publicly available.\footnote{\url{https://github.com/Shekhale/gbnns_dim_red}}
Also, in concurrent research, we show that with some additional modifications \textsc{dim-red} can be used to obtain a new state-of-the-art graph-based NNS algorithm. 

%Note that the aim of the current paper was not to develop a new algorithm but rather to theoretically analyze the existing techniques. However, the results obtained with   are quite promising. 

\section{Conclusion}

In this paper, we theoretically analyze the performance of graph-based NNS, specifically focusing on the dense regime $d \ll \log(n)$. We make an important step from practice to theory: in addition to plain NN graphs, we also analyze the effect of two heuristics widely used in practice: shortcut edges and beam search. 

Since graph-based NNS algorithms become extremely popular nowadays, we believe that more theoretical analysis of such methods will follow.
A natural direction for future research would be to find broader conditions under which similar guarantees can be obtained.
In particular, in the current research, we assume uniform distribution. We supported this assumption empirically, but clearly, the analysis of more general distributions would be useful. While our results can be straightforwardly extended to distributions similar to uniform (e.g., if the density varies in some bounded interval), further generalizations seem to be tricky.
Another promising direction is to analyze the effect of diversification~\cite{harwood2016fanng}, which was empirically shown to improve the quality of graph-based NNS~\cite{lin2019graph}. However, this is reasonable to do in a more general setting, since diversification is especially important for non-uniform distributions.

\section*{Acknowledgments}

The authors thank Artem Babenko, Dmitry Baranchuk, and Stanislav Morozov for fruitful discussions.

\newpage

\bibliographystyle{icml2020}
\bibliography{neighbors}

\begin{thebibliography}{47}
\providecommand{\natexlab}[1]{#1}
\providecommand{\url}[1]{\texttt{#1}}
\expandafter\ifx\csname urlstyle\endcsname\relax
  \providecommand{\doi}[1]{doi: #1}\else
  \providecommand{\doi}{doi: \begingroup \urlstyle{rm}\Url}\fi

\bibitem[Andoni \& Indyk(2008)Andoni and Indyk]{andoni2008near}
Andoni, A. and Indyk, P.
\newblock Near-optimal hashing algorithms for near neighbor problem in high
  dimension.
\newblock \emph{Communications of the ACM}, 51\penalty0 (1):\penalty0 117--122,
  2008.

\bibitem[Andoni \& Razenshteyn(2015)Andoni and Razenshteyn]{andoni2015optimal}
Andoni, A. and Razenshteyn, I.
\newblock Optimal data-dependent hashing for approximate near neighbors.
\newblock In \emph{Proceedings of the forty-seventh annual ACM symposium on
  Theory of computing}, pp.\  793--801. ACM, 2015.

\bibitem[Andoni \& Razensteyn(2016)Andoni and Razensteyn]{andoni2015tight}
Andoni, A. and Razensteyn, I.
\newblock Tight lower bounds for data-dependent locality-sensitive hashing.
\newblock In \emph{32nd International Symposium on Computational Geometry (SoCG
  2016)}, 2016.

\bibitem[Andoni et~al.(2017)Andoni, Laarhoven, Razenshteyn, and
  Waingarten]{andoni2017optimal}
Andoni, A., Laarhoven, T., Razenshteyn, I., and Waingarten, E.
\newblock Optimal hashing-based time-space trade-offs for approximate near
  neighbors.
\newblock In \emph{Proceedings of the Twenty-Eighth Annual ACM-SIAM Symposium
  on Discrete Algorithms}, pp.\  47--66. Society for Industrial and Applied
  Mathematics, 2017.

\bibitem[Arya et~al.(1998)Arya, Mount, Netanyahu, Silverman, and
  Wu]{arya1998optimal}
Arya, S., Mount, D.~M., Netanyahu, N.~S., Silverman, R., and Wu, A.~Y.
\newblock An optimal algorithm for approximate nearest neighbor searching fixed
  dimensions.
\newblock \emph{Journal of the ACM (JACM)}, 45\penalty0 (6):\penalty0 891--923,
  1998.

\bibitem[Aum{\"u}ller et~al.(2019)Aum{\"u}ller, Bernhardsson, and
  Faithfull]{aumuller2019ann}
Aum{\"u}ller, M., Bernhardsson, E., and Faithfull, A.
\newblock {ANN}-benchmarks: a benchmarking tool for approximate nearest
  neighbor algorithms.
\newblock \emph{Information Systems}, 2019.

\bibitem[Babenko \& Lempitsky(2016)Babenko and Lempitsky]{Babenko_2016_CVPR}
Babenko, A. and Lempitsky, V.
\newblock Efficient indexing of billion-scale datasets of deep descriptors.
\newblock In \emph{The IEEE Conference on Computer Vision and Pattern
  Recognition (CVPR)}, 2016.

\bibitem[Baranchuk et~al.(2019)Baranchuk, Persiyanov, Sinitsin, and
  Babenko]{baranchuk2019learning}
Baranchuk, D., Persiyanov, D., Sinitsin, A., and Babenko, A.
\newblock Learning to route in similarity graphs.
\newblock In \emph{International Conference on Machine Learning}, pp.\
  475--484, 2019.

\bibitem[Barri{\`e}re et~al.(2001)Barri{\`e}re, Fraigniaud, Kranakis, and
  Krizanc]{barriere2001efficient}
Barri{\`e}re, L., Fraigniaud, P., Kranakis, E., and Krizanc, D.
\newblock Efficient routing in networks with long range contacts.
\newblock In \emph{International Symposium on Distributed Computing}, pp.\
  270--284. Springer, 2001.

\bibitem[Beaumont et~al.(2007)Beaumont, Kermarrec, and
  Rivi{\`e}re]{beaumont2007peer}
Beaumont, O., Kermarrec, A.-M., and Rivi{\`e}re, {\'E}.
\newblock Peer to peer multidimensional overlays: approximating complex
  structures.
\newblock In \emph{International Conference On Principles Of Distributed
  Systems}, pp.\  315--328. Springer, 2007.

\bibitem[Becker et~al.(2016)Becker, Ducas, Gama, and Laarhoven]{becker2016new}
Becker, A., Ducas, L., Gama, N., and Laarhoven, T.
\newblock New directions in nearest neighbor searching with applications to
  lattice sieving.
\newblock In \emph{Proceedings of the twenty-seventh annual ACM-SIAM symposium
  on Discrete algorithms}, pp.\  10--24, 2016.

\bibitem[Bentley(1975)]{bentley1975multidimensional}
Bentley, J.~L.
\newblock Multidimensional binary search trees used for associative searching.
\newblock \emph{Communications of the ACM}, 18\penalty0 (9):\penalty0 509--517,
  1975.

\bibitem[Beygelzimer et~al.(2006)Beygelzimer, Kakade, and
  Langford]{beygelzimer2006cover}
Beygelzimer, A., Kakade, S., and Langford, J.
\newblock Cover trees for nearest neighbor.
\newblock In \emph{Proceedings of the 23rd international conference on Machine
  learning}, pp.\  97--104, 2006.

\bibitem[Bishop(2006)]{bishop2006pattern}
Bishop, C.~M.
\newblock \emph{Pattern recognition and machine learning}.
\newblock Springer, 2006.

\bibitem[Bonnet et~al.(2007)Bonnet, Kermarrec, and Raynal]{bonnet2007small}
Bonnet, F., Kermarrec, A.-M., and Raynal, M.
\newblock Small-world networks: is there a mismatch between theory and
  practice?
\newblock 2007.

\bibitem[Chen \& Shah(2018)Chen and Shah]{chen2018explaining}
Chen, G.~H. and Shah, D.
\newblock Explaining the success of nearest neighbor methods in prediction.
\newblock \emph{Foundations and Trends{\textregistered} in Machine Learning},
  10\penalty0 (5-6):\penalty0 337--588, 2018.

\bibitem[Dasgupta \& Freund(2008)Dasgupta and Freund]{dasgupta2008random}
Dasgupta, S. and Freund, Y.
\newblock Random projection trees and low dimensional manifolds.
\newblock In \emph{STOC}, volume~8, pp.\  537--546. Citeseer, 2008.

\bibitem[Dasgupta \& Sinha(2015)Dasgupta and Sinha]{dasgupta2015randomized}
Dasgupta, S. and Sinha, K.
\newblock Randomized partition trees for nearest neighbor search.
\newblock \emph{Algorithmica}, 72\penalty0 (1):\penalty0 237--263, 2015.

\bibitem[Datar et~al.(2004)Datar, Immorlica, Indyk, and
  Mirrokni]{datar2004locality}
Datar, M., Immorlica, N., Indyk, P., and Mirrokni, V.~S.
\newblock Locality-sensitive hashing scheme based on p-stable distributions.
\newblock In \emph{Proceedings of the twentieth annual symposium on
  Computational geometry}, pp.\  253--262. ACM, 2004.

\bibitem[Dong et~al.(2011)Dong, Moses, and Li]{dong2011efficient}
Dong, W., Moses, C., and Li, K.
\newblock Efficient k-nearest neighbor graph construction for generic
  similarity measures.
\newblock In \emph{Proceedings of the 20th international conference on World
  wide web}, pp.\  577--586. ACM, 2011.

\bibitem[Duchon et~al.(2006)Duchon, Hanusse, Lebhar, and
  Schabanel]{duchon2006could}
Duchon, P., Hanusse, N., Lebhar, E., and Schabanel, N.
\newblock Could any graph be turned into a small-world?
\newblock \emph{Theoretical Computer Science}, 355\penalty0 (1):\penalty0
  96--103, 2006.

\bibitem[Fu et~al.(2019)Fu, Xiang, Wang, and Cai]{fu2019fast}
Fu, C., Xiang, C., Wang, C., and Cai, D.
\newblock Fast approximate nearest neighbor search with the navigating
  spreading-out graph.
\newblock \emph{Proceedings of the VLDB Endowment}, 12\penalty0 (5):\penalty0
  461--474, 2019.

\bibitem[Hajebi et~al.(2011)Hajebi, Abbasi-Yadkori, Shahbazi, and
  Zhang]{hajebi2011fast}
Hajebi, K., Abbasi-Yadkori, Y., Shahbazi, H., and Zhang, H.
\newblock Fast approximate nearest-neighbor search with k-nearest neighbor
  graph.
\newblock In \emph{Twenty-Second International Joint Conference on Artificial
  Intelligence}, 2011.

\bibitem[Har-Peled et~al.(2012)Har-Peled, Indyk, and
  Motwani]{har2012approximate}
Har-Peled, S., Indyk, P., and Motwani, R.
\newblock Approximate nearest neighbor: Towards removing the curse of
  dimensionality.
\newblock \emph{Theory of computing}, 8\penalty0 (1):\penalty0 321--350, 2012.

\bibitem[Harwood \& Drummond(2016)Harwood and Drummond]{harwood2016fanng}
Harwood, B. and Drummond, T.
\newblock Fanng: Fast approximate nearest neighbour graphs.
\newblock In \emph{Proceedings of the IEEE Conference on Computer Vision and
  Pattern Recognition}, pp.\  5713--5722, 2016.

\bibitem[Indyk \& Motwani(1998)Indyk and Motwani]{indyk1998approximate}
Indyk, P. and Motwani, R.
\newblock Approximate nearest neighbors: towards removing the curse of
  dimensionality.
\newblock In \emph{Proceedings of the thirtieth annual ACM symposium on Theory
  of computing}, pp.\  604--613. ACM, 1998.

\bibitem[Iwasaki \& Miyazaki(2018)Iwasaki and
  Miyazaki]{iwasaki2018optimization}
Iwasaki, M. and Miyazaki, D.
\newblock Optimization of indexing based on k-nearest neighbor graph for
  proximity search in high-dimensional data.
\newblock \emph{arXiv preprint arXiv:1810.07355}, 2018.

\bibitem[Jegou et~al.(2010)Jegou, Douze, and Schmid]{sift_gist}
Jegou, H., Douze, M., and Schmid, C.
\newblock Product quantization for nearest neighbor search.
\newblock \emph{IEEE transactions on pattern analysis and machine
  intelligence}, 33\penalty0 (1):\penalty0 117--128, 2010.

\bibitem[Karbasi et~al.(2015)Karbasi, Ioannidis, and
  Massouli{\'e}]{karbasi2015small}
Karbasi, A., Ioannidis, S., and Massouli{\'e}, L.
\newblock From small-world networks to comparison-based search.
\newblock \emph{IEEE Transactions on Information Theory}, 61\penalty0
  (6):\penalty0 3056--3074, 2015.

\bibitem[Keivani \& Sinha(2018)Keivani and Sinha]{keivani2018improved}
Keivani, O. and Sinha, K.
\newblock Improved nearest neighbor search using auxiliary information and
  priority functions.
\newblock In \emph{International Conference on Machine Learning}, pp.\
  2578--2586, 2018.

\bibitem[Kleinberg(2000)]{kleinberg2000small}
Kleinberg, J.
\newblock The small-world phenomenon: an algorithmic perspective.
\newblock In \emph{Proceedings of the thirty-second annual ACM symposium on
  Theory of computing}, pp.\  163--170, 2000.

\bibitem[Laarhoven(2018)]{laarhoven2018graph}
Laarhoven, T.
\newblock Graph-based time-space trade-offs for approximate near neighbors.
\newblock In \emph{34th International Symposium on Computational Geometry (SoCG
  2018)}, 2018.

\bibitem[Lin \& Zhao(2019)Lin and Zhao]{lin2019graph}
Lin, P.-C. and Zhao, W.-L.
\newblock Graph based nearest neighbor search: promises and failures.
\newblock \emph{arXiv. org}, 2019.

\bibitem[Malkov et~al.(2014)Malkov, Ponomarenko, Logvinov, and
  Krylov]{malkov2014approximate}
Malkov, Y., Ponomarenko, A., Logvinov, A., and Krylov, V.
\newblock Approximate nearest neighbor algorithm based on navigable small world
  graphs.
\newblock \emph{Information Systems}, 45:\penalty0 61--68, 2014.

\bibitem[Malkov \& Yashunin(2018)Malkov and Yashunin]{malkov2018efficient}
Malkov, Y.~A. and Yashunin, D.~A.
\newblock Efficient and robust approximate nearest neighbor search using
  hierarchical navigable small world graphs.
\newblock \emph{IEEE Transactions on Pattern Analysis and Machine
  Intelligence}, 2018.

\bibitem[May \& Ozerov(2015)May and Ozerov]{may2015computing}
May, A. and Ozerov, I.
\newblock On computing nearest neighbors with applications to decoding of
  binary linear codes.
\newblock In \emph{Annual International Conference on the Theory and
  Applications of Cryptographic Techniques}, pp.\  203--228. Springer, 2015.

\bibitem[Meiser(1993)]{meiser1993point}
Meiser, S.
\newblock Point location in arrangements of hyperplanes.
\newblock \emph{Information and Computation}, 106\penalty0 (2):\penalty0
  286--303, 1993.

\bibitem[Motwani et~al.(2007)Motwani, Naor, and Panigrahy]{motwani2007lower}
Motwani, R., Naor, A., and Panigrahy, R.
\newblock Lower bounds on locality sensitive hashing.
\newblock \emph{SIAM Journal on Discrete Mathematics}, 21\penalty0
  (4):\penalty0 930, 2007.

\bibitem[Navarro(2002)]{navarro2002searching}
Navarro, G.
\newblock Searching in metric spaces by spatial approximation.
\newblock \emph{The VLDB Journal}, 11\penalty0 (1):\penalty0 28--46, 2002.

\bibitem[O'Donnell et~al.(2014)O'Donnell, Wu, and Zhou]{o2014optimal}
O'Donnell, R., Wu, Y., and Zhou, Y.
\newblock Optimal lower bounds for locality-sensitive hashing (except when q is
  tiny).
\newblock \emph{ACM Transactions on Computation Theory (TOCT)}, 6\penalty0
  (1):\penalty0 5, 2014.

\bibitem[Pennington et~al.(2014)Pennington, Socher, and Manning]{glove}
Pennington, J., Socher, R., and Manning, C.~D.
\newblock {GloVe}: global vectors for word representation.
\newblock In \emph{Proceedings of the 2014 conference on empirical methods in
  natural language processing (EMNLP)}, pp.\  1532--1543, 2014.

\bibitem[Sablayrolles et~al.(2018)Sablayrolles, Douze, Schmid, and
  J{\'e}gou]{sablayrolles2018spreading}
Sablayrolles, A., Douze, M., Schmid, C., and J{\'e}gou, H.
\newblock Spreading vectors for similarity search.
\newblock In \emph{International Conference on Learning Representations}, 2018.

\bibitem[Shakhnarovich et~al.(2006)Shakhnarovich, Darrell, and
  Indyk]{shakhnarovich2006nearest}
Shakhnarovich, G., Darrell, T., and Indyk, P.
\newblock \emph{Nearest-neighbor methods in learning and vision: theory and
  practice (neural information processing)}.
\newblock The MIT press, 2006.

\bibitem[Valiant(2015)]{valiant2015finding}
Valiant, G.
\newblock Finding correlations in subquadratic time, with applications to
  learning parities and the closest pair problem.
\newblock \emph{Journal of the ACM (JACM)}, 62\penalty0 (2):\penalty0 13, 2015.

\bibitem[Wang et~al.(2012)Wang, Wang, Zeng, Tu, Gan, and Li]{wang2012scalable}
Wang, J., Wang, J., Zeng, G., Tu, Z., Gan, R., and Li, S.
\newblock Scalable {$k$-NN} graph construction for visual descriptors.
\newblock In \emph{2012 IEEE Conference on Computer Vision and Pattern
  Recognition}, pp.\  1106--1113, 2012.

\bibitem[Watts \& Strogatz(1998)Watts and Strogatz]{watts1998collective}
Watts, D.~J. and Strogatz, S.~H.
\newblock Collective dynamics of ‘small-world’networks.
\newblock \emph{Nature}, 393:\penalty0 440–442, 1998.

\bibitem[Wu et~al.(2008)Wu, Kumar, Quinlan, Ghosh, Yang, Motoda, McLachlan, Ng,
  Liu, Philip, et~al.]{wu2008top}
Wu, X., Kumar, V., Quinlan, J.~R., Ghosh, J., Yang, Q., Motoda, H., McLachlan,
  G.~J., Ng, A., Liu, B., Philip, S.~Y., et~al.
\newblock Top 10 algorithms in data mining.
\newblock \emph{Knowledge and information systems}, 14\penalty0 (1):\penalty0
  1--37, 2008.

\end{thebibliography}


\begin{thebibliography}{1}

\bibitem{becker2016new}
A.~Becker, L.~Ducas, N.~Gama, and T.~Laarhoven.
\newblock New directions in nearest neighbor searching with applications to
  lattice sieving.
\newblock In {\em Proceedings of the twenty-seventh annual ACM-SIAM symposium
  on Discrete algorithms}, pages 10--24, 2016.

\bibitem{bollobas1988diameter}
B.~Bollob{\'a}s and F.~R.~K. Chung.
\newblock The diameter of a cycle plus a random matching.
\newblock {\em SIAM Journal on discrete mathematics}, 1(3):328--333, 1988.

\bibitem{hamza1995smallest}
K.~Hamza.
\newblock The smallest uniform upper bound on the distance between the mean and
  the median of the binomial and poisson distributions.
\newblock {\em Statistics \& Probability Letters}, 23(1):21--25, 1995.

\bibitem{laarhoven2018graph}
T.~Laarhoven.
\newblock Graph-based time-space trade-offs for approximate near neighbors.
\newblock In {\em 34th International Symposium on Computational Geometry (SoCG
  2018)}, 2018.

\bibitem{levina2005lid}
E.~Levina and P.~J. Bickel.
\newblock Maximum likelihood estimation of intrinsic dimension.
\newblock In {\em Advances in neural information processing systems}, pages
  777--784, 2005.

\bibitem{penrose1999k}
M.~D. Penrose.
\newblock On k-connectivity for a geometric random graph.
\newblock {\em Random Structures \& Algorithms}, 15(2):145--164, 1999.

\bibitem{penrose2016connectivity}
M.~D. Penrose et~al.
\newblock Connectivity of soft random geometric graphs.
\newblock {\em The Annals of Applied Probability}, 26(2):986--1028, 2016.

\bibitem{watts1998collective}
D.~J. Watts and S.~H. Strogatz.
\newblock Collective dynamics of ‘small-world’networks.
\newblock {\em Nature}, 393:440–442, 1998.

\end{thebibliography}
	
\end{document}

\maketitle

\section{Analysis of spherical caps}\label{app::sec:caps}

In this section, we formulate some technical results on the volumes of spherical caps and their intersections, which we extensively use in the proofs. Although they are similar to those formulated in~\cite{becker2016new}, it is crucial for our problem that parameters defining spherical caps depend on $d$ and may tend to zero (both in dense and sparse regimes), while the results in~\cite{becker2016new} hold only for fixed parameters. Also, in Lemma~\ref{app::lem:W}, we extend the corresponding result from~\cite{becker2016new}, as discussed further in this section.
	
\subsection{Volumes of spherical caps}\label{app::sec:lem1}

Let us denote by $\mu$ the Lebesgue measure over $\R^{d+1}$. By $C_\x(\gamma)$ we denote a spherical cap of height $\gamma$ centered at $\x \in \S^d$, i.e., $\{\y \in \S^d: \langle \x, \y \rangle \ge \gamma\}$; $C(\gamma) = \mu\left(C_\x(\gamma)\right)$ denotes the volume of a spherical cap of height $\gamma$. Recall that throughout the paper for any variable $\gamma$, $0 \le \gamma \le 1$, we let $\hat \gamma := \sqrt{1 - \gamma^2}$. 
We prove the following lemma.

\begin{lemma}\label{app::lem:C}
Let $\gamma = \gamma(d)$ be such that $0 \le \gamma \le 1$. Then
\[
\Theta\left( d^{-1/2} \right)  \hat\gamma^{d} \le C\left(\gamma\right) \le \Theta\left( d^{-1/2} \right)  \hat\gamma^{d} \min \left\{d^{1/2}, \frac{1}{\gamma} \right\}\,.
\]
\end{lemma} 

\begin{proof}
	
In order to have similar reasoning with the proof of Lemma~\ref{app::lem:W}, we consider any two-dimensional plane containing the vector $\x$ defining the cap $C_{\x}(\gamma)$ and let $p$ denote the orthogonal projection from $\S^{d}$ to this two-dimensional plane. 
	
The first steps of the proof are similar to those in~\cite{becker2016new} (but note that we analyze $\S^d$ instead of $S^{d-1}$, which leads to slightly simpler expressions). Consider any measurable subset $U$ of the two-dimensional unit ball, then the volume of the preimage $p^{-1}(U)$ (relative to the volume of $\S^{d}$) is:

\begin{equation*}
I(U) 
= \int_{r,\phi \in U} \frac{\mu(S^{d-2})}{\mu(S^{d})} \left(\sqrt{1-r^2}\right)^{d-3}r \, dr \, d\phi \,.
\end{equation*}

We define a function $g(r) = \int_{\phi: (r,\phi) \in U} d\phi$, then we can rewrite the integral as
\begin{equation*}
I(U) 
= \frac{(d-1)}{4\,\pi} \int_{0}^1 \left(1-r^2\right)^{(d-3)/2} g(r)\, dr^2\,.  
\end{equation*}

Let $U = p\left(C_{\x}(\gamma)\right)$, then, using $t = \left(1 - r^2\right)/\hat\gamma^2$, where $\hat\gamma = \sqrt{1 - \gamma^2}$, we get 
\begin{equation}\label{app::eq:int_general}
C(\gamma) = \frac{(d-1)}{4\,\pi} \int_{\gamma}^1 \left(1-r^2\right)^{(d-3)/2} g(r)\, dr^2  \\
=  \frac{(d-1)\,\hat\gamma^{d-1}}{4\,\pi} \int_{0}^{1}  g\left(\sqrt{1-\hat\gamma^2t}\right) t^{(d-3)/2}\, dt\,.
\end{equation}
	
Note that from Equation~\eqref{app::eq:int_general} we get thvolume of a hemisphere is $C(0) = 1/2$, since $g(r) = \pi$ for all $r$ in this case and $\hat \gamma = 1$. 

Now we consider an arbitrary $\gamma \ge 0$ and note that $g(r) = 2\arccos(\gamma/r)$ (see Figure~\ref{app::fig:lem1}). So, we obtain

\begin{figure}
\centering
\includegraphics[width = 0.4\textwidth]{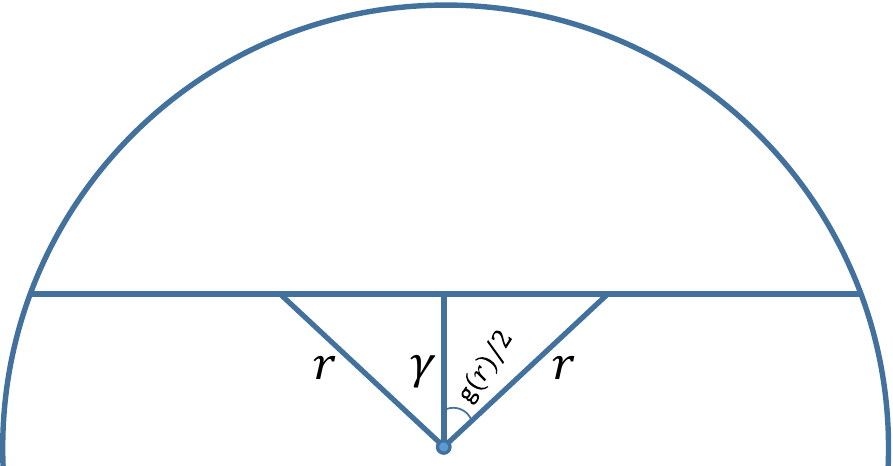}
\caption{$g(r)$}
\label{app::fig:lem1}
\end{figure}

\begin{multline*}
C(\gamma) = 
\frac{(d-1)\,\hat\gamma^{d-1}}{2\,\pi} \int_{0}^{1}  \arccos\left(\frac{\gamma}{\sqrt{1-\hat\gamma^2t}}\right) t^{\frac{d-3}{2}}\, dt 
\\
= \frac{(d-1)\,\hat\gamma^{d-1}}{2\,\pi} \int_{0}^{1}  \arcsin\left(\hat\gamma\sqrt{\frac{1-t}{1-\hat\gamma^2t}}\right) t^{(d-3)/2}\, dt \,.
\end{multline*}

Now we note that $x \le \arcsin(x) \le x\cdot\pi/2$ for $0 \le x \le 1$, so
\[
C(\gamma) = \Theta\left( d \right) \hat\gamma^{d} \int_{0}^{1}  \sqrt{\frac{1-t}{1-\hat\gamma^2 t}}\cdot  t^{(d-3)/2}\,dt\,.
\]
	
Finally, we estimate
\begin{equation}\label{app::eq:bounds}
\sqrt{1-t} \le \sqrt{\frac{1-t}{1-\hat\gamma^2 t}} \le \min\left\{1,\sqrt{\frac{1-t}{1-\hat\gamma^2}}\right\}.
\end{equation}
	
So, the lower bound is 
\begin{equation*}
C(\gamma) \ge  \Theta\left( d \right) \hat\gamma^{d}\, \mathrm{B}\left( \frac 3 2 , \frac {d-1}{2} \right) \\ = 
\Theta\left( d \right) \hat\gamma^{d} \left(\frac {d-1}{2} \right)^{-3/2} = 
\Theta\left( d^{-1/2} \right) \hat\gamma^{d} \,.
\end{equation*}
	
The upper bounds are
\[
C(\gamma) \le \Theta\left( d \right) \hat\gamma^{d} \int_{0}^{1}   t^{(d-3)/2}\,dt = \Theta\left( 1 \right) \hat\gamma^{d}\,,
\]
\begin{equation*}
C(\gamma) \le \Theta\left( d \right) \hat\gamma^{d} \int_{0}^{1} \sqrt{\frac{1-t}{1-\hat\gamma^2}} \cdot t^{(d-3)/2}\,dt \\ = \Theta\left( d^{-1/2} \right) \frac{\hat\gamma^{d}}{\gamma} \,.
\end{equation*}
This completes the proof.
\end{proof}

\subsection{Volumes of intersections of spherical caps}\label{app::sec:lem2}

By $W_{\x,\y}(\alpha,\beta)$ we denote the intersection of two spherical caps centered at $\x \in \S^d$ and $\y \in \S^d$ with heights $\alpha$ and $\beta$, respectively, i.e., $W_{\x,\y}(\alpha,\beta) = \{\z \in \S^d: \langle \z,\x \rangle \ge \alpha, \langle \z,\y \rangle \ge \beta \}$. As for spherical caps, by $W(\alpha,\beta,\theta)$ we denote the volume of such intersection given that the angle between $\x$ and $\y$ is $\theta$. 

We analyze the volume of the intersection of two spherical caps $C_{\x}(\alpha)$ and $C_{\y}(\beta)$. In the lemma below we assume $\gamma \le 1$. However, it is clear that if $\gamma > 1$, then either the caps do not intersect (if $\alpha > \beta \cos \theta$ and $\beta > \alpha \cos \theta$) or the larger cap contains the smaller one.

\begin{figure*}
\centering
\begin{subfigure}{0.45\textwidth}
\includegraphics[width=0.95\textwidth]{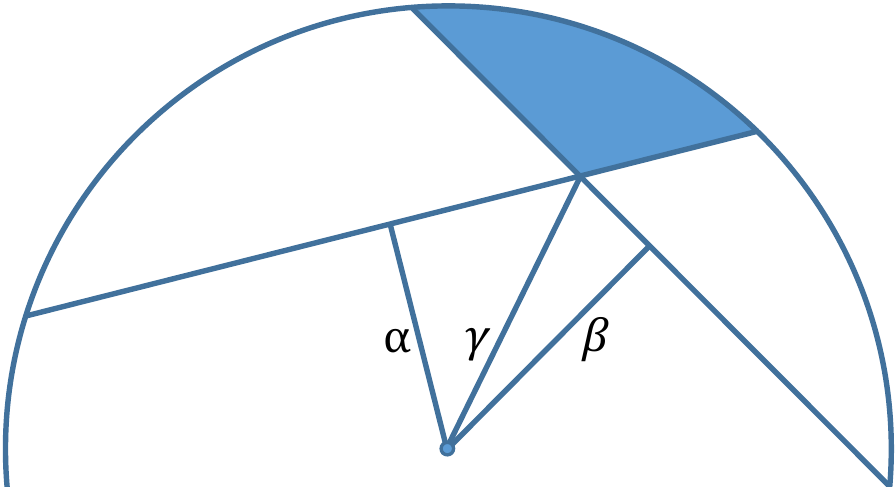}{\subcaption{$\beta > \alpha \cos \theta$, $\alpha > \beta \cos \theta$}\label{fig:1}}
\end{subfigure}
\begin{subfigure}{0.45\textwidth}
\includegraphics[width=0.93\textwidth]{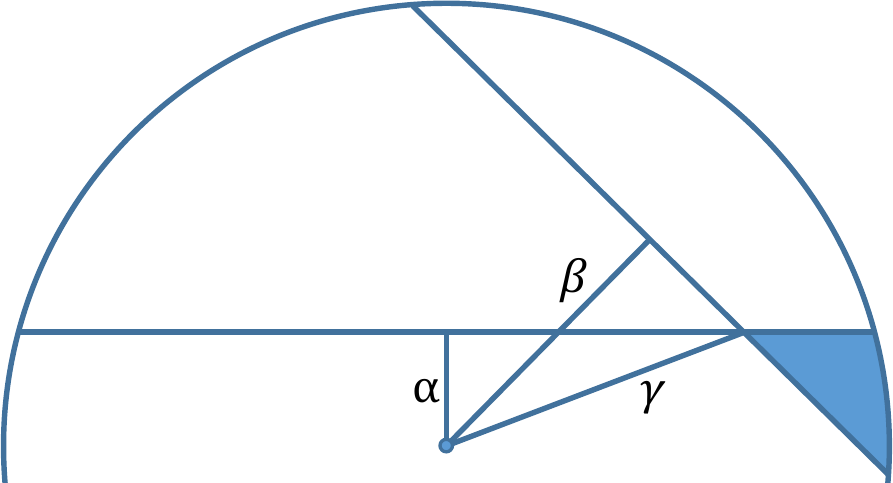}{\subcaption{$\alpha < \beta \cos \theta$}\label{fig:2}}
\end{subfigure}
\caption{Intersection of spherical caps}
\label{fig:intersection}
\end{figure*}	

\begin{lemma}\label{app::lem:W}
Let 
$\gamma = \frac{\sqrt{\alpha^2 + \beta^2 - 2 \alpha\beta\cos\theta}}{\sin\theta}$ and assume that $\gamma \le 1$, then:
\begin{enumerate}
\item[(1)] If $\alpha \le \beta \cos \theta$, then 
$
C(\beta)/2 < W(\alpha,\beta,\theta) \le C(\beta)$ and
\begin{equation*}
 C(\beta) - W(\alpha,\beta,\theta)
 \\ \le 
C_{u,\beta} \,
 \Theta\left( d^{-1} \right)  \hat\gamma^{d} \min \left\{d^{1/2}, \frac{1}{\gamma} \right\}\,;
\end{equation*}
\item[(2)] If $\beta \le \alpha \cos \theta$, then 
$C(\alpha)/2 < W(\alpha,\beta,\theta) \le C(\alpha)$ and
\begin{equation*}
C(\alpha) - W(\alpha,\beta,\theta) 
\\ \le C_{u,\alpha} \,
 \Theta\left( d^{-1} \right)  \hat\gamma^{d} \min \left\{d^{1/2}, \frac{1}{\gamma} \right\}\,;
\end{equation*}
\item[(3)] Otherwise,
\begin{equation*}
(C_{l,\alpha}+C_{l,\beta}) \, \Theta\left( d^{-1} \right)  \hat\gamma^{d} \le W(\alpha,\beta,\theta) 
\\ \le (C_{u,\alpha}+C_{u,\beta}) \, \Theta\left( d^{-1} \right)  \hat\gamma^{d} \min \left\{d^{1/2}, \frac{1}{\gamma} \right\}\,;
\end{equation*}
\end{enumerate}
where
\begin{equation*}
C_{l,\alpha} = \frac {\alpha\left(\hat\alpha \, \sin \theta - | \beta - \alpha \, \cos \theta |  \right)} {\gamma \,\hat\gamma \, \sin \theta}\,, \hspace{5pt}
C_{l,\beta} = \frac { \beta\left(\hat\beta \, \sin \theta - | \alpha - \beta \, \cos \theta |  \right)} {\gamma \,\hat\gamma \, \sin \theta}\,,
\end{equation*}
\begin{equation*}
C_{u,\alpha} = \frac{\hat\gamma\,\alpha\,\sin \theta}{\gamma |\beta - \alpha\,\cos\theta |}\,,   \hspace{5pt}
C_{u,\beta} = \frac{\hat\gamma\,\beta\,\sin \theta}{\gamma |\alpha - \beta\,\cos\theta |} \,.
\end{equation*}
\end{lemma}

The cases considered in this lemma are illustrated in Figure~\ref{fig:intersection}.

This lemma differs from Lemma 2.2 in~\cite{becker2016new} by, first, allowing the parameters $\alpha,\beta,\theta$ depend on $d$ and, second, considering the cases (1) and (2), which are essential for the proofs. Namely, we use the lower bound in (3) to show that with high probability we can make a step of the algorithm since the intersection of some spherical caps is large enough (Figure~\ref{fig:1}); we use the upper bounds in (1) and (2) to show that at the final step of the algorithm we can find the nearest neighbor with high probability, since the volume of the intersection of some spherical caps is very close to the volume of one of them (Figure~\ref{fig:2}), see the details further in the proof. 

\begin{proof}
	
Consider the plane formed by the the vectors $\x$ and $\y$ defining the caps and let $p$ denote the orthogonal projection to this plane. Let $U = p(W_{\x,\y}(\alpha,\beta))$. 
	
Denote by $\gamma$ the distance between the origin and the intersection of chords bounding the projections of spherical caps. One can show that
\[
\gamma = \sqrt{ \frac{\alpha^2 + \beta^2 - 2\alpha \beta \cos{\theta}}{\sin^2 \theta}}\,.
\]

If $\alpha \le \beta \cos \theta$, it is easy to see that $W(\alpha,\beta,\theta) > \frac 1 2 C(\beta)$, since more than a half of $C_{\y}(\beta)$ is covered by the intersection (see Figure~\ref{fig:2}). Similarly, if $\beta \le \alpha \cos \theta$, then $W(\alpha,\beta,\theta) > \frac 1 2 C(\alpha)$. Now we move to the proof of (3) and will return to (1) and (2) after that.

If $\cos \theta < \frac \alpha \beta$ and $\cos \theta < \frac \beta \alpha$, then we are in the situation shown on Figure~\ref{fig:1} and the distance between the intersection of spherical caps and the origin is $\gamma$. As in the proof of Lemma~\ref{app::lem:C}, denote $g(r) = \int_{\phi: (r,\phi) \in U} d\phi$, then the relative volume of $p^{-1}(U)$ is (see Equation~\eqref{app::eq:int_general}):
\begin{equation*}
W(\alpha,\beta,\theta) = \frac{(d-1)\,\hat\gamma^{d-1}}{4\,\pi} \int_{0}^{1}  g\left(\sqrt{1-\hat\gamma^2t}\right) t^{\frac{d-3}{2}}\, dt\,.
\end{equation*}
	
\begin{figure*}
\centering
\begin{subfigure}{0.45\textwidth}
	\includegraphics[width=0.95\textwidth]{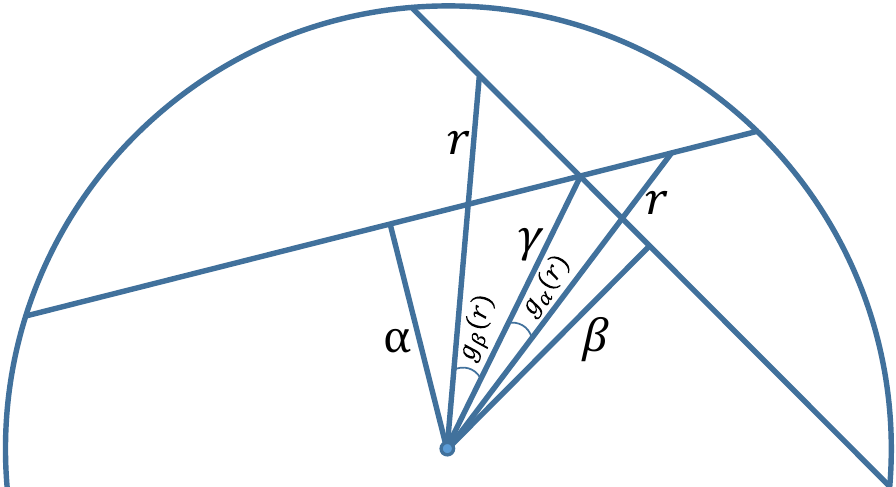}{\subcaption{$\beta > \alpha \cos \theta$, $\alpha > \beta \cos \theta$}\label{app::fig:lem2_1}}
\end{subfigure}
	\begin{subfigure}{0.45\textwidth}
		\includegraphics[width=0.95\textwidth]{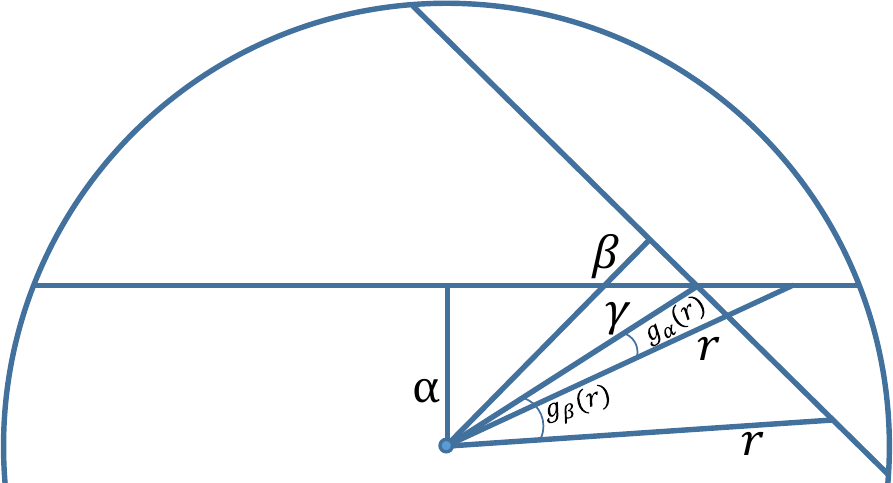}{\subcaption{$\alpha < \beta \cos \theta$}\label{app::fig:lem2_2}}
	\end{subfigure}
	\caption{$g_{\alpha}(r)$ and $g_{\beta}(r)$}
	\label{app::fig:lem2}
\end{figure*}
	
The function $g(r)$ can be written as $g_{\alpha}(r)+g_{\beta}(r)$, where (see Figure~\ref{app::fig:lem2_1})
\begin{align*}
g_{\alpha}(r) &= \arccos\left(\frac{\alpha}{r}\right)-\arccos\left(\frac{\alpha}{\gamma}\right),\\
g_{\beta}(r) &= \arccos\left(\frac{\beta}{r}\right)-\arccos\left(\frac{\beta}{\gamma}\right).
\end{align*}
Accordingly, we can write $W(\alpha,\beta,\theta) = W_{\alpha}(\alpha,\beta,\theta) + W_{\beta}(\alpha,\beta,\theta)$.
	
Let us estimate $g_\alpha\left(\sqrt{1-\hat\gamma^2t}\right)$:
\begin{multline*}
g_\alpha\left(\sqrt{1-\hat\gamma^2t}\right)
 = \arcsin\left(\sqrt{1 - \frac{\alpha^2}{1 -\hat\gamma^2t}}\right) - \arcsin\left(\sqrt{1 - \frac{\alpha^2}{\gamma^2}}\right)
\\
= \arcsin \left(\sqrt{\left(1 - \frac{\alpha^2}{1 - \hat\gamma^2t}\right) \frac{\alpha^2}{\gamma^2}} \sqrt{\left(1 - \frac{\alpha^2}{\gamma^2}\right) \frac{\alpha^2}{1 - \hat\gamma^2t}} \right) 
\\
= \Theta \left( \frac{\alpha\sqrt{\gamma^2-\alpha^2}}{\gamma\sqrt{1-\hat\gamma^2t}} \right) \left(\sqrt{1+\frac{\hat\gamma^2(1-t)}{\gamma^2-\alpha^2}}- 1\right).
\end{multline*}
	
Note that 
\begin{equation*}
\left(\sqrt{1+\frac{\hat\gamma^2}{\gamma^2-\alpha^2}}- 1\right)(1-t) \\
\le \sqrt{1+\frac{\hat\gamma^2(1-t)}{\gamma^2-\alpha^2}}- 1 \le \frac{\hat\gamma^2}{2\left(\gamma^2-\alpha^2\right)}(1-t) \,.
\end{equation*}
	
Now, we can write the lower bound for $W_{\alpha}(\alpha,\beta,\theta)$. Let 
\begin{equation*}
C_{l,\alpha} =  \left(\sqrt{1+\frac{\hat\gamma^2}{\gamma^2-\alpha^2}}- 1\right)\frac{\alpha\sqrt{\gamma^2-\alpha^2}}{\gamma\,\hat\gamma} \\
= \frac {\alpha\left(\hat\alpha \, \sin \theta - | \beta - \alpha \, \cos \theta | \right)} {\gamma \,\hat\gamma \, \sin \theta}\,,
\end{equation*}
$C_{l,\beta}$ can be obtained by swapping $\alpha$ and $\beta$. 

Then the lower bound is
\begin{multline*}
W(\alpha,\beta,\theta) 
\ge \Theta(d)\,\hat\gamma^{d} \left(C_{l,\alpha} + C_{l,\beta}\right) \int_{0}^{1} \frac{1-t}{\sqrt{1-\hat\gamma^2t}}   t^{(d-3)/2}\, dt 
\\
\ge
\Theta(d)\,\hat\gamma^{d} \left(C_{l,\alpha} + C_{l,\beta}\right) \int_{0}^{1} (1-t)\, t^{(d-3)/2}\, dt
= \Theta(d^{-1})\,\hat\gamma^{d} \left(C_{l,\alpha} + C_{l,\beta}\right).
\end{multline*}
	
Now we define $C_{u,\alpha}$ (and, similarly, $C_{u,\beta}$) as
\[
C_{u,\alpha} = \frac{\hat\gamma}{\left(\gamma^2-\alpha^2\right)} \cdot \frac{\alpha\sqrt{\gamma^2-\alpha^2}}{\gamma} 
= \frac{\hat\gamma\,\alpha\,\sin \theta}{\gamma\,|\beta - \alpha\,\cos\theta |} \,.
\]
Then
\begin{equation*}
W(\alpha,\beta,\theta) \le \Theta(d)\,\hat\gamma^{d} \left(C_{u,\alpha} + C_{u,\beta}\right) \\ \cdot\int_{0}^{1} \frac{1-t}{\sqrt{1-\hat\gamma^2t}}   t^{(d-3)/2}\, dt\,.
\end{equation*}
	
We use the upper bound
\[
\frac{1-t}{\sqrt{1-\hat\gamma^2t}} \le \min \left\{\sqrt{1-t}, \frac{1-t}{\gamma}\right\}
\]
and obtain
\[
W(\alpha,\beta,\theta) \le
\Theta(d^{-1})\,\hat\gamma^{d} \left(C_{u,\alpha} + C_{u,\beta}\right) \min
\left\{d^{\frac 1 2}, \frac{1}{\gamma}\right\},
\]
which completes the proof of (3).
	
Now, let us finish the proof for (1) and (2). If $\alpha \le \beta \cos \theta$, then we are in a situation shown on Figure~\ref{fig:2}. In this case, the bounds on $W(\alpha, \beta, \theta)$ are obvious. To estimate $C(\beta) - W(\alpha, \beta, \theta)$, we can directly follow the above proof for (3) and the only difference would be that $g(r) = g_{\beta}(r) - g_{\alpha}(r)$ instead of $g(r) = g_{\alpha}(r) + g_{\beta}(r)$. 
Note that we need only the upper bound and we simply say that $g(r) \le g_{\beta}(r)$.
The proof for (2) is similar with $g(r) = g_{\alpha}(r) - g_{\beta}(r)$. 
\end{proof}

\section{Greedy search on plain NN graphs}

\subsection{Proof overview}\label{sec:general_idea}
	
Let $\alpha_M$ denote the height of a spherical cap defining $G(M)$. By $f = f(n) = (n-1)C(\alpha_M)$ we denote the expected number of neighbors of a given node in $G(M)$. Then, it is clear that the complexity of one step of graph-based search is $\Theta\left( f \cdot d\right)$ (with high probability), so for making $k$ steps we need $\Theta\left(k \cdot f \cdot d\right)$ computations (see Section~\ref{app::sec:time}). The number of edges in the graph is $\Theta\left(f\cdot n\right)$, so the space complexity is $\Theta\left(f\cdot n \cdot \log n\right)$ (see Section~\ref{app::sec:space}).  
	
To prove that the algorithm succeeds, we have to show that it does not get stuck in a local optimum until we are sufficiently close to $q$. If we take some point $\x$ with $\langle \x,q \rangle = \alpha_s$, then the probability of making a step towards $q$ is determined by $W(\alpha_M, \alpha_s, \arccos{\alpha_s})$. In all further proofs we obtain lower bounds for this value of the form $\frac{1}{n} g(n)$ with  $1 \ll g(n) \ll n$. From this, we easily get that the probability of making a step is at least $1 - \left(1 - g(n)/n\right)^{n-1} = 1 - e^{-g(n)(1+o(1))}$.
	
A fact that will be useful in the proofs is that the value $W(\alpha_M,\alpha_s,\arccos\alpha_s)$ is a monotone function of $\alpha_s$ (see Section~\ref{app::sec:monotone}). I.e., if we have a lower bound for some $\alpha_s$, then for all smaller values we have this bound automatically. 
	
By estimating the value $W(\alpha_M, \alpha_s, \arccos{\alpha_s})$, we obtain (in further sections) that with probability $1-o(1)$ we reach some point at distance at most $\arccos{\alpha_s}$ from $q$. Then, to achieve success, we may either jump directly to $\bar \x$ at the next step or to already have $\arccos{\alpha_s} \le c R$ if we are solving $c,R$-ANN.
	
To limit the number of steps, we additionally show that with a sufficiently large probability at each step we become ``$\varepsilon$ closer'' to $q$. In the dense regime, it means that the sine of the angle between the current position and $q$ becomes smaller by at least some fixed value.
	
Let us emphasize that several consecutive steps of the algorithm cannot be analyzed independently. Indeed, if at some step we moved from $\x$ to $\y$, then there were no points in $C_{\x}(\alpha_M)$ closer to $q$ than $\y$ by the definition of the algorithm. Consequently, the intersection of $C_{\x}(\alpha_M)$, $C_{\y}(\alpha_M)$ and $C_{q}\left(\langle q,\y \rangle \right)$ contains no elements of the dataset. The closer $\y$ to $\x$ the larger this intersection. However, the fact that at each step we become at least ``$\varepsilon$ closer'' to $q$ allows us to bound the volume of this intersection and to prove that it can be neglected.

This is worth noting that in the proofs below we assume that the elements are distributed according to the Poisson point process on $\S^{d}$ with $n$ being the \textit{expected} number of elements. This makes the proofs more concise without changing the results since the distributions are asymptotically equivalent. Indeed, conditioning on the number of nodes in the Poisson process, we get the uniform distribution, and the number of nodes in the Poisson process is $\Theta(n)$ with high probability.

\subsection{Auxiliary results}

\subsubsection{Time complexity}\label{app::sec:time}

Let $v$ be an arbitrary node of $G$ and let $N(v)$ denote the number of its neighbors in $G$. Recall that $f = (n-1)C(\alpha_M)$.

\begin{lemma}\label{app::lem:one-step}
With probability at lest $1-\frac{4}{f}$ we have $\frac 1 2 f \le N(v) \le \frac{3}{2}f$.
\end{lemma}

\begin{proof}
The number of neighbors $N(v)$ of a node $v$ follows Binomial distribution  $\mathrm{Bin}(n-1,C(\alpha_M))$, so $\E N(v) = f$. From Chebyshev's inequality we get
\[
\P\left(|N(v) - f| > \frac{f}{2}\right) \le \frac{4 \, \Var (N(v))}{f^2} \le \frac{4}{f},
\]
which completes the proof.
\end{proof}

To obtain the final time complexity of graph-based NN search, we have to sum up the complexities of all steps of the algorithm. We obtain the following result.

\begin{lemma}\label{app::lem:k-steps}
If we made $k$ steps of the graph-based NNS, then with probability $1 - O\left(\frac{1}{k\, f}\right)$ the obtained time complexity is $\Theta\left( k f d \right)$.
\end{lemma}

\begin{proof}
Although the nodes encountered in one iteration are not independent, the fact that we do not need to measure the distance from any point to $q$ more than once allows us to upper bound the complexity by the random variable distributed according to $\mathrm{Bin}(k(n-1), C(\alpha_M))$. Then, we can follow the proof of Lemma~\ref{app::lem:one-step} and note that one distance computation takes $\Theta(d)$.

To see that the lower bound is also $\Theta\left( k f d \right)$, we note that more than a constant number of steps are needed only for the dense regime. For this regime, we may follow the reasoning of Lemma~\ref{app::lem:dependence} to show that the volume of the intersection of two consecutive balls is negligible compared to the volume of each of them.
\end{proof}

\subsubsection{Space complexity}\label{app::sec:space}

\begin{lemma}\label{app::lem:space}
With probability $1 - O\left(\frac{1}{f\,n}\right)$ we have $\frac 1 4 \, f \, n \le E(G) \le \frac 3 4 \, f \, n$.
\end{lemma}

\begin{proof}
	
The proof is straightforward.\footnote{Similar proof appeared in, e.g.,~\cite{laarhoven2018graph}.} For each pair of nodes, the probability that there is an edge between them equals $C\left(\alpha_M\right)$. Therefore, the expected number of edges is 
\[
\E(E(G)) = \binom{n}{2} C\left(\alpha_M\right)\,.
\]
	
It remains to prove that $E(G)$ is tightly concentrated near its expectation. For this, we apply Chebyshev's inequality, so we have to estimate the variance $\Var(E(G))$. One can easily see that if we are given two pairs of nodes $e_1$ and $e_2$, then, if they are not the same (while one coincident node is allowed), then $\P(e_1,e_2 \in E(G)) = C(\alpha_M)^2$. Therefore, 
\begin{multline*}
\Var(E(G)) = \hspace{-10pt} \sum_{\hspace{7pt}e_1, e_2 \in \binom{\D}{2}}\hspace{-5pt}\P(e_1,e_2 \in E(G)) \ - \left(\E E(G)\right)^2  
\\
= \sum_{\substack{e_1, e_2 \in \binom{\D}{2}\\ e_1 \neq e_2}}\P(e_1,e_2 \in E(G))  + \E E(G) - \left(\E E(G)\right)^2 
 = \binom{n}{2}C\left(\alpha_M\right)\left(1-C\left(\alpha_M\right)\right)\,.
\end{multline*}
Applying Chebyshev's inequality, we get
\begin{equation}\label{app::eq:space}
\P\left(|E(G) - \E(E(G))| > \frac{E(G)}{2} \right) \le \frac{4\,\Var(E(G))}{E(G)^2} \\ = \frac{4\,(1 - C(\alpha))}{E(G)}\,.
\end{equation}
	
From this, the lemma follows. 
\end{proof}

It remains to note that if we store a graph as the adjacency lists, then the space complexity is $\Theta\left(E(G) \cdot \log n \right)$.

\subsubsection{Monotonicity of $W(\alpha_M,\alpha_s,\arccos\alpha_s)$}\label{app::sec:monotone}

\begin{lemma}
$W(\alpha_M,\alpha_s,\arccos\alpha_s)$ is a non-increasing function of $\alpha_s$.
\end{lemma}

\begin{proof}
	
\begin{figure}
\centering
\includegraphics[width = 0.4 \textwidth]{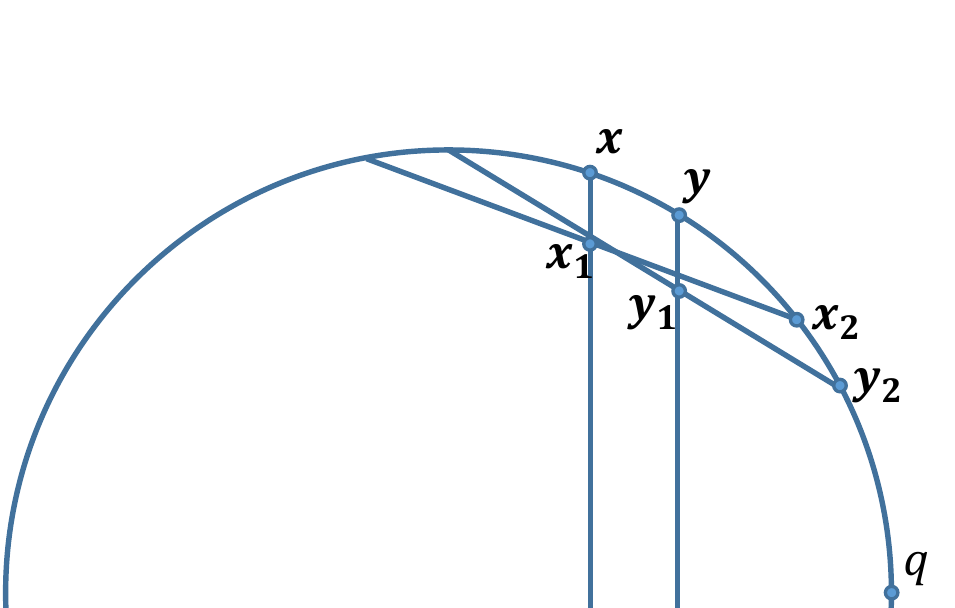}
\caption{Monotonicity of $W(\alpha_M,\alpha_s,\arccos\alpha_s)$}
\label{app::fig:monotone}
\end{figure}
	
We refer to Figure~\ref{app::fig:monotone}, where two spherical caps of height $\alpha_M$ are centered at $\x$ and $\y$, respectively, and note that we have to compare ``curved triangles'' $\triangle \,\x\, \x_1\, \x_2$ and $\triangle \,\y\,\y_1\,\y_2$. Obviously, $\rho(\x,\x_2) = \rho(\y,\y_2)$, $\angle \, \x\,\x_2\,\x_1 = \angle \, \y\,\y_2\,\y_1$, but $\angle \, \x\,\x_1\,\x_2 < \angle \, \y\,\y_1\,\y_2$. From this and the spherical symmetry of $\mu(p^{-1}(\cdot))$ ($p$ was defined in the proof of Lemma~\ref{app::lem:W}) the result follows.
\end{proof}

\subsection{Proof of Theorem~1 (greedy search in dense regime)}
%Theorem~\ref{thm:dense_main}}
	
Recall that for dense datasets ($d = \log n / \omega$), it is convenient to operate with radii of spherical caps (if $\alpha$ is a height of a spherical cap, then we say that $\hat \alpha$ is its radius). Let $\hat\alpha_1$ be the radius of a cap centered at a given point and covering its nearest neighbor, then we have $C(\alpha_1) \sim \frac 1 n$, i.e., $\hat\alpha_1 \sim n^{-\frac{1}{d}} = 2^{-\omega}$. We further let $\delta:= 2^{-\omega}$. 

We construct a graph $G(M)$ using spherical caps with radius $\hat\alpha_M = M \, \delta$.  Then, from Lemma~\ref{app::lem:C}, we get $f = \Theta\left( n\,d^{-1/2} M^d \delta^{d} \right) = \Theta\left(d^{-1/2} M^d\right)$. So, the number of edges in $G(M)$ is $\Theta\left(d^{-1/2}\cdot M^d\cdot n \right)$ and the space complexity is $\Theta\left(d^{-1/2} \cdot M^d\cdot n \cdot \log n\right)$ (see Section~\ref{app::sec:space}).
	
Let us now analyze the distance $\arccos \alpha_s$ up to which we can make steps towards the query $q$ (with sufficiently large probability). 
This is stated in Lemma~\ref{app::lem:varepsilon}, but before that let us prove some auxiliary results.

First, let us analyze the behavior of the main term $\hat \gamma^d$ in $W(x,y,\arccos z)$ when $\hat x,\hat y,\hat z = o(1)$, which is the case for the considered situations in the dense regime. 

\begin{lemma}\label{app::lem:dense_gen}
If $\hat x,\hat y,\hat z = o(1)$, then $\hat \gamma$ defined in Lemma~\ref{app::lem:W} for $W(x,y,\arccos z)$ is
\[
\hat\gamma \sim \frac{\sqrt{2\left(\hat x^2\hat y^2 + \hat y^2 \hat z^2 + \hat x^2 \hat z^2\right) - \left(\hat x^4 + \hat y^4 + \hat z^4 \right)}}{2 \hat z}.
\]
\end{lemma}

\begin{proof}
	
By the definition, $\gamma = \sqrt{\frac{x^2 + y^2 - 2 xyz}{1-z^2}}$. Then 
	
\begin{multline*}
\hat\gamma^2 = 1 - \gamma^2  = \frac{\hat x^2 + \hat y^2 + \hat z^2 - 2 + 2 \sqrt{\left(1-\hat x^2\right)\left(1-\hat y^2\right)\left(1-\hat z^2\right)}}{\hat z^2} 
\\
\sim \frac{2\left(\hat x^2\hat y^2 + \hat y^2 \hat z^2 + \hat x^2 \hat z^2\right) - \left(\hat x^4 + \hat y^4 + \hat z^4 \right)}{4 \hat z^2}\,.
\end{multline*}
\end{proof}

Now, we analyze $W(\alpha_M, \alpha_s, \arccos \alpha_s)$ and we need only the lower bound. Recall that we use the notation $\delta = 2^{-\omega}$.

\begin{lemma}\label{app::lem:dense_s}
Assume that $\hat\alpha_s = s \,\delta$ and $\hat\alpha_M = M \, \delta$.  
\begin{itemize}
\item If $M \ge \sqrt{2} s$, then $W(\alpha_M, \alpha_s, \arccos \alpha_s) \ge \frac{1}{n}\cdot s^{d+o(d)}$.
\item If $M < \sqrt{2} s$, then $W(\alpha_M, \alpha_s, \arccos \alpha_s) \ge \frac{1}{n} \cdot \left(M^2 - \frac{M^4}{4 s^2} \right)^{d/2 + o(d)}$.
\end{itemize}
\end{lemma}

\begin{proof}
First, assume that $M \ge \sqrt{2} s$. In this case we have $\alpha_M < \alpha_s^2$, so we are under the conditions (1)-(2) of Lemma~\ref{app::lem:W} (see Figure~\ref{fig:2}) and, using Lemma~\ref{app::lem:C}, we get $W(\alpha_M, \alpha_s, \arccos \alpha_s) > \frac{1}{2} C(\alpha_s) = \Theta\left(d^{-1/2} s^d \delta^d\right) = \frac{1}{n}s^{d+o(d)}$.
	
If $M < \sqrt{2}s$, then, asymptotically, we have $\alpha_M > \alpha_s^2$, so the case (3) of Lemma~\ref{app::lem:W} can be applied. Let us use Lemma~\ref{app::lem:dense_gen} to estimate $\hat \gamma$:
\[
\hat \gamma^2 = \delta^2 \left( M^2 -  \frac{M^4}{4 s^2}\right)\left(1+o(1)\right) \,.
\]
And now from Lemma~\ref{app::lem:W} we get
\[
W(\alpha_M, \alpha_s, \arccos \alpha_s) \ge C_l \,  \Theta\left(d^{-1}\right)\hat\gamma^d\,,
\]
where $C_l$ corresponds to the sum of $C_{l,\alpha}$ and $C_{l,\beta}$ in Lemma~\ref{app::lem:W}. So, it remains to estimate $C_{l}$:
\begin{multline*}
C_{l} = \frac {\alpha_M\left(\hat\alpha_M \, \hat\alpha_s + \alpha_M \alpha_s - \alpha_s  \right) + \alpha_s\left(\hat\alpha_s^2  + \alpha_s^2 - \alpha_M  \right)} {\gamma \,\hat\gamma \, \hat \alpha_s}  \\
= \Theta  \left(\frac{Ms\delta^2 + \sqrt{(1-M^2\delta^2)(1-s^2\delta^2)} - \sqrt{1-s^2\delta^2} + s^2\delta^2 + 1 - s^2\delta^2 - \sqrt{1 - M^2\delta^2}}{\delta^2}\right) \\
= \Theta(1) \cdot \frac{M(s-M/2)\delta^2 + \frac{1}{2}M^2\delta^2}{\delta^2} = \Theta(1)\,.
\end{multline*}
Therefore,
\[
W(\alpha_M, \alpha_s, \arccos \alpha_s) \ge \frac{1}{n} \cdot \left(M^2 - \frac{M^4}{4 s^2} \right)^{d/2 + o(d)}.
\]
\end{proof}

From Lemma~\ref{app::lem:dense_s}, we can find the conditions for $M$ and $s$ to guarantee (with sufficiently large probability) making steps until we are in the cap of radius $s \delta$ centered at $q$. The following lemma gives such conditions and also guarantees that at each step we can reach a cap of a radius at least $\varepsilon \delta$ smaller for some constant $\varepsilon > 0$.

\begin{lemma}\label{app::lem:varepsilon}
Assume that $s>1$. If $M > \sqrt{2}s$ or $M^2 - \frac{M^4}{4 s^2} > 1$, then there exists such constant $\varepsilon > 0$ that $W(\alpha_M, \alpha_s, \arcsin\left(\hat\alpha_s + \varepsilon\delta\right)) \ge \frac{1}{n} S^{d(1+o(1))}$ for some constant $S > 1$.
\end{lemma}

\begin{proof}
First, let us take $\varepsilon = 0$. Then, the result directly follows from Lemma~\ref{app::lem:dense_s}. We note that a value $M$ satisfying $M^2 - \frac{M^4}{4 s^2} > 1$ exists only if $s>1$. 
	
Now, let us demonstrate that we can take some $\varepsilon > 0$. The two cases discussed in Lemma~\ref{app::lem:dense_s} now correspond to $M^2 \ge s^2 + (s+\varepsilon)^2$ and  $M^2 < s^2 + (s+\varepsilon)^2$, respectively. If $M > \sqrt{2} s$, then we can choose a sufficiently small $\varepsilon$ that $M^2 \ge s^2 + (s+\varepsilon)^2$. Then, the result follows from Lemma~\ref{app::lem:dense_s} and the fact that $s>1$. Otherwise, we have $M^2 < s^2 + (s+\varepsilon)^2$ and instead of the condition $M^2 - \frac{M^4}{4 s^2} > 1$ we get (using Lemma~\ref{app::lem:dense_gen})
\[
\frac{2\left(M^2 s^2 + s^2 (s+\varepsilon)^2 + M^2 (s+\varepsilon)^2\right) - \left(M^4 + s^4 + (s+\varepsilon)^4 \right)}{4 (s+\varepsilon)^2} > 1.
\]
As this holds for $\varepsilon = 0$, we can choose a small enough $\varepsilon > 0$ that the condition is still satisfied. 
\end{proof}
	
This lemma implies that if we are given $M$ and $s$ satisfying the above conditions, then we can make a step towards $q$ since the expected number of nodes in the intersection of spherical caps is much larger than 1. Formally, we can estimate from below the values $g(n)$ for all steps of the algorithm by $g_{min}(n) = S^{d(1+o(1))}$. So, according to Section~\ref{sec:general_idea}, we can make each step with probability $1 - O\left(e^{-{S^{d(1+o(1))}}}\right)$. Moreover, each step reduces the radius of a spherical cap centered at $q$ and containing the current position by at least $\varepsilon \delta$. As a result, the number of steps (until we reach some distance $\arccos \alpha_s$) is $O\left(\delta^{-1}\right) = O\left(2^{\omega}\right)$. 

To estimate the overall success probability, we have to take into account that the consecutive steps of the algorithm are dependent. 
In Section~\ref{sec:general_idea}, it is explained how the previous steps of the algorithm may affect the current one: the fact that at some step we moved from $\y$ to $\x$ implies that there were no elements closer to $q$ than $\x$ in a spherical cap centered at $\y$. However, we can show that this dependence can be neglected.

\begin{lemma}\label{app::lem:dependence}
The dependence of the consecutive steps can be neglected and does not affect the analysis.
\end{lemma}

\begin{proof}

\begin{figure}
	\centering
	\includegraphics[width = 0.45 \textwidth]{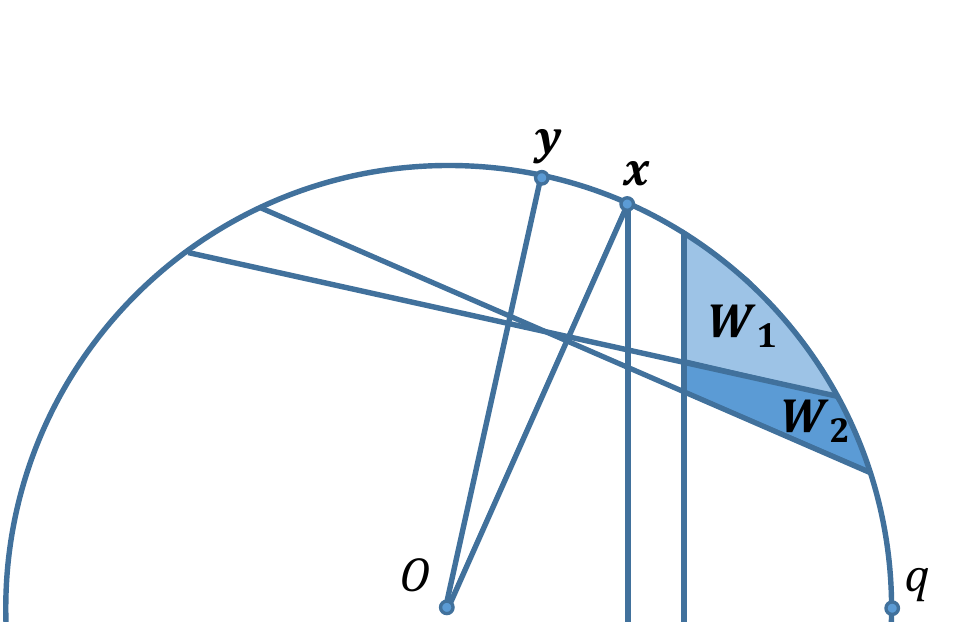}
	\caption{Dependence of consecutive steps}
	\label{app::fig:dependence}
\end{figure}

The main idea is illustrated on Figure~\ref{app::fig:dependence}. Assume that we are currently at a point $\x$ with $\rho(\x,q) = \arcsin\left(\hat\alpha_s + \varepsilon\delta\right)$. Then, as in the proof of Lemma~\ref{app::lem:varepsilon}, we are interested in the volume $W(\alpha_M, \alpha_s, \arcsin\left(\hat\alpha_s + \varepsilon\delta\right))$, which corresponds to $W_1 + W_2$ on Figure~\ref{app::fig:dependence}. Assume that at the previous step we were at some point $\y$. Given that all steps are ``longer than $\varepsilon$'', the largest effect of the previous step is reached when $\y$ is as close to $\x$ as possible and $\x$ lies on the geodesic between $\y$ and $q$. Therefore, the largest possible volume is $W(\alpha_M, \alpha_s, \arcsin(\hat\alpha_s + 2\varepsilon\delta))$, which corresponds to $W_1$ on Figure~\ref{app::fig:dependence}. It remains to show that $W_1$ is negligible compared with $W_1 + W_2$.

If $M < \sqrt{2}s$, then the main term of $W(\alpha_M, \alpha_s, \arcsin\left(\hat\alpha_s + \varepsilon\delta\right))$ is $\hat\gamma^d$ with
\begin{multline*}
\hat\gamma^2 = \frac{2\left(M^2 s^2 + s^2 (s+\varepsilon)^2 + M^2 (s+\varepsilon)^2\right) - \left(M^4 + s^4 + (s+\varepsilon)^4 \right)}{4 (s+\varepsilon)^2} \\
= \frac {M^2 + s^2} 2  
- \frac{\left(M^2 - s^2\right)^2}{4 (s+\varepsilon)^2} - \frac{(s+\varepsilon)^2}{4}\,.
\end{multline*}

The main term of $W(\alpha_M, \alpha_s, \arcsin(\hat\alpha_s + 2\varepsilon\delta))$ is $\hat\gamma_1^d$ with 
\[
\hat\gamma_1^2 = \frac {M^2 + s^2} 2  
- \frac{\left(M^2 - s^2\right)^2}{4 (s+2\varepsilon)^2} - \frac{(s+2\varepsilon)^2}{4}\,.
\]
It is easy to see that $M < \sqrt{2}s$ implies that $\hat\gamma_1^2 < \hat\gamma^2$. As a result, $W(\alpha_M, \alpha_s, \arcsin(\hat\alpha_s + 2\varepsilon\delta)) = o\left(W(\alpha_M, \alpha_s, \arcsin(\hat\alpha_s + \varepsilon\delta))\right)$ (similarly to the other proofs, it is easy to show that the effect of the other terms in Lemma~\ref{app::lem:W} is negligible compared to $\left(\hat\gamma_1/\hat \gamma\right)^d$). Moreover, since at each step we reduce the radius of a spherical cap centered at $q$ by at least $\varepsilon\delta$, any cap encountered in one iteration intersects with only a constant number of other caps, so their overall effect is negligible, which completes the proof.

Finally, let us note that as soon as we have $M > \sqrt{2}s$ (with $s > 1$), with probability $1-o(1)$ we find the nearest neighbor in one step.
\end{proof}

Having this result on  ``almost independence'' of the consecutive steps, we can say that the overall success probability is $1 - O\left(2^{\omega}e^{-{S^{d+o(1)}}}\right)$. Assuming $d \gg \log \log n$, we get $O\left(2^{\omega}e^{-{S^{d(1+o(1))}}}\right) =  O\left(2^{\frac{\log n}{d}}e^{-\log n}\right) = o(1)$. This concludes the proof for the success probability $1-o(1)$ up to choosing suitable values for $s$ and $M$.
	
Let us discuss the time complexity. With probability $1-o(1)$ the number of steps is $\Theta\left(\delta^{-1}\right)$: the upper bound was already discussed; the lower bound follows from the fact that $M\delta$ is the radius of a spherical cap, so we cannot make steps longer than $\arcsin(M\delta)$, and with probability $1-o(1)$ we start from a constant distance from $q$. The complexity of each step is $\Theta\left(f\cdot d \right) = \Theta\left(d^{1/2} \cdot M^d \right)$, so overall we get $\Theta\left(d^{1/2} \cdot 2^{\omega} \cdot M^d \right)$. 
	
It remains to find suitable values for $s$ and $M$.
Before we continue, let us analyze the conditions under which we find exactly the nearest neighbor at the next step of the algorithm. 
Assume that the radius of a cap centered at $q$ and covering the currently considered element is $\hat\alpha_s$ and $\hat\alpha_s = s\,\delta$, $\hat\alpha_M = M \delta$. Further assume that the radius of a spherical cap covering points at distance at most $R$ from $q$ is $\hat\alpha_r = r\delta = \sin R$ for some $r < 1$. The following lemma gives the conditions for $M$ and $s$ such that at the next step of the algorithm we find the nearest neighbor $\bar \x$ with probability $1 - o(1)$ given that $\bar \x$ is uniformly distributed within a distance $R$ from $q$.

\begin{lemma}\label{lem:dense_exact}
If for constant $M,s,r$ we have $M^2 > s^2 + r^2$, then 
\[
C(\alpha_r) - W(\alpha_M, \alpha_r, \arccos \alpha_s) \le C(\alpha_r) \beta^d
\]
with some $\beta < 1$.
\end{lemma}

\begin{proof}
	
First, recall the lower bound for $C(\alpha_r)$: $C(\alpha_r) \ge \Theta\left(d^{-1/2}\right) \delta^{d}r^d$. 
	
Since we have $M^2 > s^2 + r^2$, then asymptotically we have $\alpha_M < \alpha_s \alpha_r$, so the cases (1)-(2) of Lemma~\ref{app::lem:W} should be applied (see Figure~\ref{fig:2}). Let us estimate $\hat\gamma^2$ (Lemma~\ref{app::lem:dense_gen}):
\begin{multline*}
\hat\gamma^2 \sim \delta^2\cdot \frac{2\left(M^2s^2 + M^2 r^2 + s^2 r^2\right) - \left(M^4 + r^4 + s^4 \right)}{4\,s^2} \\
=  \delta^2\cdot \frac{-\left(M^2 - s^2 - r^2\right)^2 + 4\,s^2r^2}{4\,s^2} = \delta^2 r^2 \left(1 - \Theta(1)\right).
\end{multline*}
	
Therefore $\hat \gamma^d \le \delta^d r^d \beta^d$ with some $\beta < 1$.
	
It only remains to estimate the other terms in the upper bound from Lemma~\ref{app::lem:W}:
\[ 
\frac{\hat\gamma\,\alpha_r\,\hat\alpha_s}{\gamma |\alpha_M - \alpha_r\,\alpha_s |}
\Theta\left( d^{-1} \right) \min \left\{d^{1/2}, \frac{1}{\gamma} \right\} = O\left(d^{-1}\right)\,,
\]
from which the lemma follows.
\end{proof}

Now we are ready to finalize the proof.
We solve $c,R$-ANN if either we have $\arcsin(s\delta) \le c R$ or we are sufficiently close to $q$ to find the exact nearest neighbor $\bar \x$ in the next step of the algorithm.
Let us analyze the first possibility. Let $\sin R = r\delta$, then we need $s < c\,r$. According to Lemma~\ref{app::lem:varepsilon}, it is sufficient to have $r\,c > 1$ and either $M \ge \sqrt{2} \, rc$ or $M^2 - \frac{M^4}{4r^2c^2} > 1$.
Alternatively, according to Lemma~\ref{lem:dense_exact}, to find the exact nearest neighbor with probability $1-o(1)$, it is sufficient to reach such $s$ that $M^2 > s^2+r^2$. For this, according to Lemma~\ref{app::lem:varepsilon}, it is sufficient to have $s>1$, $M^2 > s^2 + r^2$, and either $M \ge \sqrt{2} s$ or $M^2 - \frac{M^4}{4s^2} > 1$. 

One can show that if the following conditions on $M$ and $r$ are satisfied, then we can choose an appropriate $s$ for the two cases discussed above:
\begin{itemize}[nolistsep]
\item[(a)] $r\,c > 1$ and $M^2 > 2\,r^2\,c^2 \left(1 - \sqrt{1 -  \frac{1}{r^2c^2}}\right)$;
\item[(b)]
$M^2 > \frac{2}{3} \left(r^2 + 1 + \sqrt{r^4 - r^2 + 1}\right)$.
\end{itemize}
To succeed, we need either (a) or (b) to be satisfied. The bound in (a) decreases with $r$ ($r > 1/c$) and for $r = \frac{1}{c}$ it equals $\sqrt{2}$. 
The bound in (b) increases with $r$ and for $r = 1$ it equals $\sqrt{2}$. 
To find a general bound holding for all $r$, we take the ``hardest'' $r \in (\frac 1 c, 1)$, where the bounds in (a) and (b) are equal to each other. This value is $r = \sqrt{\frac{4c^2}{(c^2+1)(3c^2-1)}}$ and it gives the bound $M > \sqrt{\frac{4c^2}{3c^2 - 1}}$ stated in the theorem.
	
\subsection{Larger neighborhood in dense regime}\label{sec:larger_neighborhood}

As discussed in the main text, it could potentially be possible that taking $M = M(n) \gg 1$ improves the query time. The following theorem shows that this is not the case.
	
\begin{theorem}\label{thm:dense_growing}
Let $M = M(n) \gg 1$. Then, with probability $1-o(1)$, graph-based NNS finds the exact nearest neighbor in one iteration; time complexity is $\Omega\left(d^{1/2}\cdot 2^{\omega} \cdot M^{d-1} \right)$; space complexity is $\Theta\left(n\cdot d^{-1/2}\cdot M^d\cdot \log n\right)$.
\end{theorem}
	
As a result, when $M \to \infty$, both time and space complexities become larger compared with constant~$M$ 
%(see Theorem~\ref{thm:dense_main}).
(see Theorem~1 from the main text).

\begin{proof}
	
When $M$ grows with $n$, it follows from the previous reasoning that the algorithm succeeds with probability $1 - o(1)$. The analysis of the space complexity is the same as for constant $M$, so we get $\Theta\left(d^{-1/2} \cdot M^d\cdot n \cdot \log n\right)$. When analyzing the time complexity, we note that the one-step complexity is $\Theta\left(d^{-1/2} \cdot M^d\right)$. It is easy to see that we cannot make steps longer than $O\left(M\cdot 2^{-\omega}\right)$. This leads to the time complexity $\Omega\left(d^{1/2} \cdot M^{d-1} \cdot 2^{\omega}\right)$.
\end{proof}

\subsection{Proof of Theorem~2 (greedy search in sparse regime)}

%~\ref{thm:sparse_main}}
	
For sparse datasets, instead of radii, we operate with heights of spherical caps. In this case, we have $C(\alpha_1) \sim \frac 1 n$, i.e., $\alpha_1^2 \sim 1 - n^{-\frac{2}{d}} = 1 - 2^{-\frac{2}{\omega}} \sim \frac{2 \ln 2}{\omega}$. We further denote $\frac{2 \ln 2}{\omega}$ by  $\delta$. 
	
We construct $G(M)$ using spherical caps with height $\alpha_M$, $\alpha_M^2 = M \, \delta$, where $M$ is some constant.  Then, from Lemma~\ref{app::lem:C}, we get that the expected number of neighbors of a node is $f = \Theta\left( n\,d^{O(1)} ( 1 -M\delta)^{d/2}\right)  = n^{1-M+o(1)}$. From this and Section~\ref{app::sec:space} the stated space complexity follows. The one-step time complexity $n^{1-M+o(1)}$ follows from Section~\ref{app::sec:time}.
	
Our aim is to solve $c,R$-ANN with some $c>1$, $R>0$. If $R \ge \pi/2c$, then we can easily find the required near neighbor within a distance $\pi/2$, since we start $G(M)$-based NNS from such point. Let us consider any $R < \frac{\pi}{2c}$. It is clear that in this case we have to find the nearest neighbor itself, since $R\,c$ is smaller than the distance to the (non-planted) nearest neighbor with probability $1+o(1)$. Note that $\alpha_c = \cos\frac{\pi}{2c} < \alpha_R := \cos{R}$.

\begin{lemma}\label{app::lem:sparse_M_s}
Assume that $\alpha_M^2 = M \delta$, $\alpha_s^2 = s \delta$, $\alpha_\varepsilon^2 = \varepsilon \delta$, $M, s>0$ are constants, and $\varepsilon \ge 0$ is bounded by a constant. If $s + M < 1$, then  
\[
W(\alpha_M, \alpha_s, \arccos \alpha_\varepsilon) \ge \frac{1}{n} \cdot  n^{\Omega(1)}\,.
\] 	
\end{lemma}

\begin{proof}

Asymptotically, we have $\alpha_M > \alpha_s \alpha_\varepsilon$ and $\alpha_s > \alpha_M \alpha_\varepsilon$, so we are under the condition (3) in Lemma~\ref{app::lem:W}.	
First, consider the main term of $W(\alpha_M,\alpha_s,\arccos\alpha_\varepsilon)$:
\[
\hat\gamma^d = \left(1 - \frac{M\delta + s \delta - 2 \sqrt{Ms\,\varepsilon\,\delta}\, \delta}{1 - \varepsilon \delta}\right)^{d/2} = e^{-\frac{d}{2}\delta \left(M + s + O\left(\sqrt{\delta}\right)  \right)} \\ = n^{-\left(M+s + O\left(\sqrt{\delta}\right)\right)}
= \frac{1}{n} \cdot  n^{\Omega(1)}.
\]
	
It remains to multiply this by $\Theta\left(d^{-1}\right)$ and $C_l = C_{l,\alpha} + C_{l,\beta}$ (see Lemma~\ref{app::lem:W}). It is easy to see that $C_l = \Omega(1)$ in this case, so both terms can be included to $n^{\Omega(1)}$, which concludes the proof.
\end{proof}

It follows from Lemma~\ref{app::lem:sparse_M_s} that if $M+s<1$, then we can reach a spherical cap with height $\alpha_s = \sqrt{s \delta}$ centered at $q$ in just one step (starting from a distance at most $\pi/2$). And we get $g(n) =  n^{\Omega(1)}.$ 

Recall that $M < \frac{\alpha_c^2}{\alpha_c^2 + 1}$ and let us take $s = \frac{1}{\alpha_c^2 + 1}$, then we have $M+s < 1$. The following lemma discusses the conditions for $M$ and $s$ such that at the next step of the algorithm we find $\bar \x$ with probability $1 - o(1)$.

\begin{lemma}\label{lem:sparse_find}
If for constant $M$ and $s$ we have $M < s \alpha_R^2$, then 
\[
C(\alpha_R) - W(\alpha_M, \alpha_R, \arccos \alpha_s) = C(\alpha_R) n^{-\Omega(1)}\,.
\]
\end{lemma}

\begin{proof}
	
First, recall the lower bound for $C(\alpha_R)$: $C(\alpha_R) \ge \Theta\left(d^{-1/2}\right) (1 - \alpha_R^2)^{d/2}$. 
	
Note that since we have $M < s \alpha_R^2$, then the cases (1)-(2) of Lemma~\ref{app::lem:W} should be applied (see Figure~\ref{fig:2}). Let us estimate $\gamma$:
\begin{multline*}
\gamma^2 = \frac{M \delta + \alpha_R^2 - 2 \sqrt{Ms}\alpha_R \delta}{1 - s \delta} 
= \alpha_R^2 + \delta\left(M - 2 \sqrt{Ms}\alpha_R + s \alpha_R^2 \right) + O\left(\delta^2\right) \\
= \alpha_R^2 + \delta \left(\sqrt{M} - \sqrt{s} \alpha_R \right)^2 + O\left(\delta^2\right) = \alpha_R^2 + \Theta(\delta)\,.
\end{multline*}
	
Therefore,
\[
\hat \gamma^d = \left(1-\alpha_R^2 \right)^{d/2}\left(1 - \Theta(\delta)\right)^{d/2} = \left(1 - \alpha_R^2\right)^{d/2} n^{-\Omega(1)}\,.
\]
	
It only remains to estimate the other terms in the upper bound from Lemma~\ref{app::lem:W}:
\[
\frac{\hat\gamma\,\alpha_R\,\hat\alpha_s}{\gamma |\alpha_M - \alpha_R\,\alpha_s |}
\Theta\left( d^{-1} \right) \min \left\{d^{1/2}, \frac{1}{\gamma} \right\} = O\left(\delta^{-1}d^{-1}\right)\,,
\]
from which the lemma follows.
\end{proof}
Note that in our case we have $M < \frac{\alpha_c^2}{\alpha_c^2 + 1} < \frac{\alpha_R^2}{\alpha_c^2 + 1} = s \alpha_R^2$. From this Theorem~2 follows.

\section{Long-range links}

\subsection{Random edges}\label{sec:random_edges}
	
The simplest way to obtain a graph with a small diameter from a given graph is to connect each node to a few random neighbors. This idea is proposed in~\cite{watts1998collective} and gives $O\left(\log n \right)$ diameter for the so-called ``small-world model''.  It was later confirmed that adding a little randomness to a connected graph makes the diameter small~\cite{bollobas1988diameter}. However, we emphasize that having a logarithmic diameter does not guarantee a logarithmic number of steps in graph-based NNS, since these steps, while being greedy in the underlying metric space, may not be optimal on a graph. 
	
To demonstrate the effect of long edges, assume that there is a graph $G'$, where each node is connected to several random neighbors by directed edges. For simplicity or reasonings, assume that we first perform NNS on $G'$ and then continue on the standard NN graph $G$. It is easy to see that during NNS on $G'$, the neighbors considered at each step are just randomly sampled nodes, we choose the one closest to $q$ and continue the process, and all such steps are independent. Therefore, the overall procedure is basically equivalent to a random sampling of a certain number of nodes and then choosing the one closest to $q$ (from which we then start the standard NNS on $G$). 
	
\begin{theorem}\label{thm:dense_random}
Under the conditions of  Theorem~1 in the main text, performing a random sampling of some number of nodes and choosing the one closest to $q$ as a starting point for graph-based NNS does not allow to get time complexity better than $\Omega\left(d^{1/2}\cdot e^{\omega(1+o(1))} \cdot M^{d}\right)$. 
\end{theorem}
	
\begin{proof}

Assume that we sample $e^{l\omega}$ nodes with an arbitrary $l = l(n)$. Then, with probability $1-o(1)$, the closest one among them lies at a distance $\Theta\left(e^{-\frac{l \omega}{d}}\right)$. As a result, the overall time complexity becomes $\Theta\left(e^{-\frac{l\omega}{d}}\cdot d^{1/2}\cdot e^{\omega} \cdot M^d + d \cdot e^{l\omega}\right)$. If $l = \Omega(d)$, then the term $d\cdot e^{l\omega}= d\cdot e^{\Omega(d)\omega}$ dominates $d^{1/2}\cdot e^{\omega} \cdot M^d$, otherwise we get $\Theta\left(d^{1/2}\cdot e^{\omega(1+o(1))} \cdot M^{d}\right)$, which proves Theorem~\ref{thm:dense_random}.
\end{proof}

\subsection{Proof of Theorem~3 (effect of proper long edges)}

%~\ref{thm:dense_kleinberg}}

Recall that we assume the following probability distribution:
\begin{equation}\label{eq:kleinberg_prob_2}
\P(\text{edge from $u$ to $v$}) = \frac{\rho(u,v)^{-d}}{\sum_{w \neq u} \rho(u, w)^{-d}}\,.
\end{equation}

First, we estimate the denominator. In the lemma below we consider only the elements $w$ with $\rho(u, w) > n^{-\frac{1}{d}}$.
However, it easily follows from the proof that adding only edges with $\rho(u, v) > n^{-\frac{1}{d}}$ does not affect the reasoning.

From Theorem~1 in the main text, we know that without long edges we need $O(n^{\frac{1}{d}})$ steps, which is less than $\log ^2 n$ for $d > \frac{\log n}{ 2 \log \log n}$. So, in this case Theorem~3 follows from Theorem~2. Hence, in the lemma below we can assume that $d < \frac{\log n}{ 2 \log \log n}$.

\begin{lemma}\label{lem:ro_expacted} 
If $d < \frac{\log n}{ 2 \log \log n}$, then
$$
\mathbb{E}\left( \rho(u,w)^{-d} \right) = \Theta\left(\frac{\log n}{\sqrt{d}}\right)  \,.
$$
\end{lemma}

\begin{proof}
Note that $\mathbb{E} \left(\rho(u,w)^{-d} \right)  = \mathbb{E} \rho^{-d }(\textbf{1}, w)$, where $\textbf{1} = (1, 0, \ldots, 0)$. So, similarly to Lemma~\ref{app::lem:C},
$$ 
\mathbb{E}\left( \rho(\textbf{1},w)^{-d} \right)
= \frac{\mu(S^{d-1})}{\mu(S^d)}  \int\limits_{-1}^{\cos(n^{-\frac{1}{d}})} 
\left(1 - x^2 \right)^{\frac{d-2}{2}}  \; (\arccos{x})^{-d} dx \,.
$$
From Stirling's approximation, we have $\frac{\mu(S^{d-1})}{\mu(S^d)} = \Theta(\sqrt{d})$. After replacing $y = \arccos{x}$, the integral becomes
$$ \int\limits_{n^{-\frac{1}{d}}}^{\pi}  y^{-d} \sin^{d - 1}(y) dy 
=
\int\limits_{n^{-\frac{1}{d}}}^{\pi} \frac{1}{y}  \left(\frac{\sin(y)}{y}\right)^{d - 1} dy 
<
\int\limits_{n^{-\frac{1}{d}}}^{\pi} \frac{1}{y} dy
=
\Theta \left( \frac{\ln n}{d} \right)\,.
$$

On the other hand, for $d < \frac{\log n}{ 2 \log \log n}$:
$$ 
\mathbb{E}\left( \rho(\textbf{1},w)^{-d} \right) 
=
\Theta(\sqrt{d})
\int\limits_{n^{-\frac{1}{d}}}^{\pi} \frac{1}{y}  \left(\frac{\sin(y)}{y}\right)^{d - 1} dy 
>
\Theta(\sqrt{d})
\int\limits_{n^{-\frac{1}{d}}}^{\frac{1}{\sqrt{d}}} \frac{1}{y}  \left(\frac{\sin(y)}{y}\right)^{d - 1} dy .
$$

Since on this interval we have $ \left(\frac{\sin(y)}{y}\right)^{d - 1} = \Theta(1)$, we can continue:

$$
\mathbb{E}\left( \rho(\textbf{1},w)^{-d} \right)  > 
\Theta(\sqrt{d})
\int\limits_{n^{-\frac{1}{d}}}^{\frac{1}{\sqrt{d}}} \frac{1}{y} dy 
=
\Theta(\sqrt{d}) \left( \frac{\ln n}{d} - \frac{1}{2} \ln d \right)
=
 \Theta \left( \frac{\ln n}{\sqrt{d}} \right)\,.
$$
As a result, we get $\mathbb{E} \rho(\textbf{1},w)^{-d}  =  \Theta \left( \frac{\log n}{\sqrt{d}} \right)\,.$
\end{proof}

Also, from the proof above it follows that $\mathbb{E} \rho(u,w)^{-2d} < \Theta(\sqrt{d}) \int\limits_{n^{-\frac{1}{d}}}^{\pi} \frac{1}{y^{d+1}} dy = O(\frac{n}{\sqrt{d}})$ 
. Let 
\[
\mathrm{Den} = \sum_{w: \rho(u, w) > n^{-\frac{1}{d}}} \rho(u, w)^{-d}, \] so
$\mathbb{E}\, \mathrm{Den} = \Theta\left( \frac{n \; \log n}{\sqrt{d}}\right)$ and 
$$
\mathbb{E} \, \mathrm{Den}^2 = n \, \mathbb{E} \rho(u,w)^{-2d} + n(n-1) \left(\mathbb{E} \rho(u,w)^{-d} \right)^2 = O\left( \frac{n^2}{\sqrt{d}} + \frac{n^2 \ln^2 n}{d}  -\frac{n \ln^2 n}{d} \right). 
$$
Finally, from Chebyshev's inequality, we get
\[
\P\left(|\mathrm{Den} - \mathbb{E} \mathrm{Den}| > \frac{\mathbb{E} \mathrm{Den}}{2}\right) \le \frac{4 \, \Var (\mathrm{Den})}{(\mathbb{E} \mathrm{Den})^2} = O \left( \frac{ \sqrt{d} }{\log^2 n} \right) = o(1).
\]

So, we may further replace the denominator of Equation~\eqref{eq:kleinberg_prob_2}  by $O\left(\frac{n \; \log n}{ \sqrt{d}}\right)$.\footnote{More formally, our analysis below is conditioned on the fact that the denominator is less than $\frac{C\, n \log n }{\sqrt d}$ for some constant $C>0$. The probability that it does not hold is $o(1)$ and for such nodes we can just assume that we do not use long edges.}

We are ready to prove the theorem.
We split the search process on a sphere into $\log n$ phases and show that each phase requires $O(\log n)$ steps. 
Phase $j$ consists of the following nodes: $\{u:  t_{j+1} < \rho(u, q) \leq t_j \}$, where $t_j  
 = \frac{\pi}{2} \; t^j
 = \frac{\pi}{2} \; \left(1 - \frac{1}{d}\right)^j$.

We start at a distance at most $\frac{\pi}{2}$, this corresponds to $j = 0$. Recall that the nearest neighbor (in the dense regime) is at a distance about $2^{-\frac{\log n}{d}}$.
Then, the number of phases needed to reach the nearest neighbor is 
\[
k \sim - \frac{1}{\log\left(1 - \frac 1 d\right)} \cdot \frac{\log n}{d} \sim \log n \,.
\]

Suppose we are at some node belonging to a phase $j$. Let us prove the following inequality for the probability of making a step to a phase with a larger number: 
$$ \P(\text{make a step to a closer phase}) >
\frac{\Theta(1)}{ \log n}\,.
$$

In the polar coordinates, we can express this probability as

\begin{multline*}
\P(\text{make a step to a closer phase}) 
=
\frac{(d-1) \; \sqrt{d} }{2 \pi  \; \log n}
\int^1_{\cos t_{i+1}} \int^{\arccos \frac{\cos(t_{i+1})}{r}}_{-\arccos \frac{\cos(t_{i+1})}{r}} 
\left(\sqrt{1 - r^2}\right)^{d-3} \\ \cdot
(\arccos(\sin(t_i) r \sin (\phi)  + \cos(t_i) r \cos (\phi)))^{-d} r \; d \phi\; dr  \\
=
\frac{\Theta(d^{\frac{3}{2}})}{ \log n}
\int^1_{\cos t_{i+1}} \int^{\arccos \frac{\cos(t_{i+1})}{r}}_{-\arccos \frac{\cos(t_{i+1})}{r}} 
\left(\sqrt{1 - r^2}\right)^{d-3}
(\arccos(r  \cos(t_i- \phi)))^{-d} r \; d \phi\; dr .
\end{multline*}

Let $r = \cos(\psi)$ and $\phi = t_j - \phi$, then the integral becomes 
$$
\int\limits^{t_{j+1}}_0
\cos \psi \left(\sin \psi \right)^{d-2}
\int\limits_{t_{j} - \arccos \frac{\cos(t_{j+1})}{\cos \psi}}^{t_{j} + \arccos \frac{\cos(t_{j+1})}{\cos \psi}}
(\arccos(\cos \psi \cos (\phi))^{-d} d \phi d \psi\,.
$$

From convexity of $\log \cos \sqrt{x}$, it follows that $\forall \psi, \phi \in [0, \frac{\pi}{2}]$ we have 
 $$\arccos(\cos \psi \cos \phi) \leq \sqrt{\psi^2 + \phi^2}\,,$$
 $$\arccos \frac{\cos(t_{j+1})}{\cos \psi} \geq \sqrt{t_{j+1}^2 - \psi^2}\,,$$
 $$ \sin(\psi) \geq \psi - \frac{\psi^3}{6} \,.$$

We use these bounds since we need a lower bound for the integral. Also, we replace the upper limit of the inner integral with $t_{j}$ and consider $\psi = t_{j} \psi$ and $\phi = t_{j} \phi$:

$$
\int\limits^{t}_0
F(t_j, \psi) \psi^{d-2}
\int\limits^{1}_{1 -\sqrt{t^2 - \psi^2}} 
 (\sqrt{\psi^2 + \phi^2})^{-d} \;d \phi \;d \psi\,,
$$
where $ F(t_{j}, \psi) = \cos(t_{j} \; \psi) \left( 1 - \frac{t_{j}^2 \psi^2}{6} \right)^{d-2}$.

Consider the inner integral:
\begin{multline*}
\int\limits^{1}_{1 -\sqrt{t^2 - \psi^2}}
(\sqrt{\psi^2 + \phi^2})^{-d} \;
\frac{1}{2 \phi} d \phi^2
>
\frac{1}{2} \int\limits^{1}_{1 - 2 \sqrt{t^2 - \psi^2} + t^2 - \psi^2} 
(\psi^2 + x)^{-\frac{d}{2}} \;d x
\\ =
\frac{1}{d - 2} \left((1 - 2 \sqrt{t^2 - \psi^2} + t^2)^{-\frac{d - 2}{2}} -(1 + \psi^2)^{-\frac{d - 2}{2}} \right)\,.
\end{multline*}

Substitute the second term to the original integral and estimate it from above:
$$
\int^{t}_0
F(t_j, \psi) \psi^{d-2}
 (1 + \psi^2)^{-\frac{d - 2}{2}}\;d \psi
\leq
\int^{t}_0
\left(\frac{\psi^2}{1 + \psi^2} \right)^{\frac{d - 2}{2}} \; d \psi
= o \left( \frac{1}{\sqrt{d}} \right). 
$$

Now we estimate from below the second term
$\int\limits^{t}_0
\left( \frac{\psi^2}{1 - 2 \sqrt{t^2 - \psi^2} + t^2} \right)^{\frac{d - 2}{2}} \;d \psi$.

Note that if $\psi = \frac{2}{\sqrt{d}}$ and $t = 1 - \frac{1}{d}$, then 
\begin{multline*}
\left( \frac{\psi^2}{1 - 2 \sqrt{t^2 - \psi^2} + t^2} \right)^{\frac{d - 2}{2}} = 
\left( \frac{\frac{4}{d}}{1 - 2 \sqrt{1 - \frac{2}{d} + \frac{1}{d^2} - \frac{4}{d}} + 1 - \frac{2}{d} + \frac{1}{d^2}} \right)^{\frac{d - 2}{2}} \\ > 
\left( \frac{\frac{4}{d}}{\frac{4}{d} + \frac{11}{d^2} + \frac{1}{d^2}} \right)^{\frac{d - 2}{2}}
=
e^{-\frac{3}{2}} (1 + o(1))\,.
\end{multline*}
Similarly, it can be shown that if $\psi = \frac{3}{\sqrt{d}}$, then

$$
\left( \frac{\psi^2}{1 - 2 \sqrt{t^2 - \psi^2} + t^2} \right)^{\frac{d - 2}{2}} > 
\Theta(1).
$$

So, for $\psi \in [\frac{2}{\sqrt{d}}, \frac{3}{\sqrt{d}}]$ this fraction is greater than some constant (as well as $F(t_j, \psi)$)  and the derivative does not change the sign on this segment. 
As a result, 

$$
\int^{t}_0
F(t_j, \psi) \left( \frac{\psi^2}{1 - 2 \sqrt{t^2 - \psi^2} + t^2} \right)^{\frac{d - 2}{2}} d \psi >  \frac{\Theta(1)}{\sqrt{d}} \,.
$$

And finally, 
$$
\P(\text{make a step in to a closer phase}) > \frac{\Theta(1)}{ \log n}\,.
$$

To sum up, there are $O(\log n)$ phases and the number of steps in each phase is geometrically distributed with the expected value $O( \log n)$. From this the theorem follows.

\subsection{Proof of Corollary~2}
%\ref{cor:log_long}}
We have
\begin{equation}\label{sec:step_prob}
\P(\text{short-cut step within $\log n$ trying}) = 1 - (1 - P)^{\log n} = \left(1 - \frac{1}{e}\right)(1 - o(1)),
\end{equation}
where $P$ is the probability corresponding to one long-range edge, which is estimated in the proof above.

Also, since $d \gg \log \log n$, we have $\frac{M^d}{\sqrt{d}} > \log n$, so the step complexity is the same.

It is easy to see that increasing the number of shortcut edges does not improve the asymptotic complexity, since the probability in~\eqref{sec:step_prob} is already constant.

\subsection{Proof of Lemma~1 (effect of pre-sampling)}\label{suppl:fastsampling}

%\ref{lem:appr_kleinberg}

For convenience, in this proof we assume that the overall number of elements is $n+1$ instead of $n$ which does not affect the analysis.

Let $v$ be the $k$-th neighbor for the source node $u$. 
For the initial distribution we have:
$$
\P(\text{edge from $u$ to $v$}) \sim   \frac{1}{k \ln n}\,.
$$

By pre-sampling of $n^{\varphi}$ nodes, we modify this probability to
\begin{equation}\label{suppl:appr}
\P( \text{edge from $u$ to $v$} | \text{$v$ is sampled}) \cdot \P( \text{$v$ is sampled}).
\end{equation}

The second term above is equal to  $\frac{n^{\varphi}}{n}$. 
Assuming that $k = n^{\alpha} > n^{1 - \varphi}$, we can estimate the probability above. Below by $l$ we denote the rank of $v$ in the selected subset and obtain:

\begin{multline*}
\P( \text{edge from $u$ to $v$} | \text{$v$ is sampled}) \cdot \P( \text{$v$ is sampled}) \\
=\frac{n^{\varphi}}{n} \sum_{l=1}^{\min(n^{\alpha},  n^{\varphi})} \binom{n^{\varphi} - 1}{l-1} \left(\frac{n^{\alpha}-1}{n - 1}\right)^{l-1}  \left(\frac{n-n^{\alpha}}{n - 1}\right)^{n^{\varphi}-l} \frac{1}{l} \; \frac{1}{\ln n^{\varphi}} 
\\
=  \frac{1}{ \varphi \; n \; \ln n} \sum_{l=1}^{\min(n^{\alpha},  n^{\varphi})} \binom{n^{\varphi}}{l} \left(\frac{n^{\alpha}-1}{n - 1}\right)^{l-1} \; \left(\frac{n-n^{\alpha}}{n - 1}\right)^{n^{\varphi}-l} \\
= \frac{\Theta(1)}{\varphi \; n^{\alpha} \ln n} \sum_{l=1}^{\min(n^{\alpha},  n^{\varphi})} \binom{n^{\varphi}}{l} \left(\frac{n^{\alpha}-1}{n - 1}\right)^{l} \; \left(\frac{n-n^{\alpha}}{n - 1}\right)^{n^{\varphi}-l} \,.
\end{multline*}

Let us analyze the sum above. First, it is easy to see that it is less than 1. 
Second, if $n^{\varphi} \le n^{\alpha}$, then the sum is ``almost equal'' to 1 (without one term corresponding to $l=0$, which we analyze below). Otherwise, we know that for a binomial distribution its median cannot lie too far away from the mean (see, e.g.,~\cite{hamza1995smallest}). Since $ \alpha > 1 - \varphi $, we have

$$\min(n^{\alpha},  n^{\varphi}) \ge \frac{n^{\varphi} (n^{\alpha} - 1)}{n - 1} + 1 =
\mathbb{E}\, \mathrm{Bin}(n^{\varphi}, \frac{n^{\alpha} - 1}{n - 1}) + 1 > 
\mathrm{median}(\varphi, \alpha).
$$

Hence,
$$
\sum_{l=0}^{\min(n^{\alpha},  n^{\varphi})} \binom{n^{\varphi}}{l} \left(\frac{n^{\alpha}-1}{n - 1}\right)^{l} \; \left(\frac{n-n^{\alpha}}{n - 1}\right)^{n^{\varphi}-l} 
>
\frac{1}{2}.
$$

Note that we added one term corresponding to $l=0$, but it is easy to see that in the worst case it is about $\frac{1}{e}$. Namely, for $l=0$:

$$
 \binom{n^{\varphi}}{0} \left(\frac{n^{\alpha}-1}{n - 1}\right)^{0} \; \left(\frac{n-n^{\alpha}}{n - 1}\right)^{n^{\varphi}-0} 
=
\left(\frac{n-n^{\alpha}}{n - 1}\right)^{n^{\varphi}} 
<
\left(1 - \frac{n^{1 - \varphi} - 1}{n - 1}\right)^{n^{\varphi}} = \frac{1}{e} (1 + o(1)).
$$

\section{Proof of Theorem~4 (effect of beam search)}
%\section{Beam search (proof of Theorem~\ref{thm:beam_search_common})}

%\textit{We believe that the requirement $M > \sqrt{4/3}$ is redundant and can be removed from the statement of the theorem. Even without this condition, it is possible to prove that with probability $1-o(1)$ the subgraph induced by the elements of the $L$-neighborhood does not have any isolated nodes. 

Let us call a spherical cap of radius $L \delta$ centered at $x$ an $L$-neighborhood of $x$.
%The most important ingredient of the proof is to show that the subgraph of $G(M)$ induced by the $L$-neighborhood of the query is connected with high probability.
We first show that the subgraph of $G(M)$ induced by the $L$-neighborhood contains a path from a given element to the nearest neighbor of the query with high probability.
%of the query is connected with high probability.

For random geometric graphs in $d$-dimensional Euclidean space (for fixed $d$) it is known that the absence of isolated nodes implies connectivity~\cite{penrose1999k,penrose2016connectivity}. However, generalizing~\cite{penrose1999k,penrose2016connectivity} to our setting is non-trivial, especially taking into account that the dimension grows as the logarithm of the number of elements in the $L$-neighborhood. 
%Therefore, we add an additional assumption for $M$ to get a simpler proof.
%Therefore, we do not use this result directly and prove the connectivity for our case. 
In our case, it is easy to show that with high probability there are no isolated nodes. Moreover, the expected degree is about $S^d$ for some $S>1$.
Hence, it is possible to prove that the graph is connected.
However, for simplicity, we prove a weaker result: for two fixed points, there is a path between them with high probability.

Let us denote by $N$ the number of nodes in the $L$-neighborhood. According to Lemma~\ref{app::lem:one-step}, with high probability, this value is $\Theta\left(d^{-1/2}L^d\right)$. So, we further assume that there are $N = \Theta\left(d^{-1/2}L^d\right)$ points uniformly distributed within the $L$-neighborhood.

Let us make the following observation that simplifies the reasoning. Consider the $L$-neighborhood of $q$. Let us project all $N$ points to the boundary of a neighborhood (moving them along the rays starting at $q$) and construct a new graph on these elements using the same $M$-neighborhoods. It is easy to see that this operation may only remove some edges and never adds new ones. Therefore, it is sufficient to prove connectivity assuming that $N$ nodes are uniformly distributed on a boundary of the $L$-neighborhood. This allows us to avoid boundary effects and simplify reasoning.

Let $p_1$ be the probability that two random nodes are connected. This probability is the volume of the $M$-neighborhood of a node normalized by the volume of the boundary of the $L$-neighborhood. Under the conditions on $M$ and $L$, one can show that $p_1$ is at least $\left(\frac{S}{L}\right)^d$ for some constant $S>1$.

We fix any pair of nodes $u,v$ and estimate the probability that there is a path of length $k$ between them. We assume that $k\to \infty$ and $k = o(\sqrt{Np_1})$, which is possible to achieve since $p_1 N \to \infty$.
%To claim connectivity, it is sufficient to show that the probability of not having such a path is $o(N^{-1})$, since for a fixed point we want to show the existence of paths to all other nodes.
We show that the probability of not having such a path is $o(1)$.

Let us denote by $P_k(u,v)$ the number of paths of length $k$ between $u$ and $v$. Then, 
\[
\E P_k(u,v) \sim \binom{N-2}{k-1} (k-1)! p_1^k \gtrsim N^{k-1} \left(\frac{S}{L}\right)^{dk} %= L^{d(k-1)}\left(\frac{S}{L}\right)^{dk}
%= \frac{1}{L^d}S^{dk} = 
= \left(\frac{N}{L^d}\right)^{k} \frac{1}{N} \,S^{dk}\,.
\]
%Roughly speaking, the expectation grows as $S^{dk}$.
We have $\E P_k(u,v) \to \infty$ if $k \to \infty$. 

To claim concentration near the expectation, we estimate the variance. 
Note that $P_k(u,v) = \E(\sum_i I_i)$, where $I_i$ indicates the event that a particular path is present and $i$ indexes all possible paths of length $k$. Then, we can estimate
\begin{multline*}
\E P_k(u,v)^2 - (\E P_k(u,v))^2  = \E \bigg(\sum_i I_i \bigg)^2  - (\E P_k(u,v))^2 \\ 
=  \sum_i \P(I_i = 1) \sum_j \left(\P(I_j = 1|I_i = 1) - \P(I_j=1) \right)  \\
= \E P_k(u,v)  \sum_j \left(\P(I_j = 1|I_i = 1) - \P(I_j = 1)
\right)\,.
%\E P_k(u,v) + \binom{n-2}{k-1}(k-1)! \sum_{j:j \neq i} \P(I_i, I_j)\,.
\end{multline*}
It is easy to see that $\sum_j \left(\P(I_j = 1|I_i = 1) - \P(I_j = 1)
\right) = o\left(\E P_k(u,v)\right)$.
Indeed, for most pairs of paths we have $\P(I_j = 1|I_i = 1) \sim \P(I_j = 1) \sim p_1^k$ since they do not share any intermediate nodes. Let us show that the contribution of the remaining pairs is small. 
The fraction of pairs of paths sharing $k_0$ intermediate nodes is $O\left(\frac{k^{2k_0}}{N^{k_0}} \right)$. Then, $\P(I_j = 1|I_i = 1) \le \P(I_j = 1)/p_1^{k_0}$, since in the worst case the paths may share $k_0$ consecutive edges. Since $k^2 \ll Np_1$, the relative contribution is $\sum_{k_0 \ge 1} O\left( \left(\frac{k^2}{Np_1}\right)^{k_0}\right) = o(1)$. Therefore, we get $\textrm{Var}(P_k(u,v)) = o\left(\left(\E P_k(u,v)\right)^2\right)$

Finally, it remains to apply Chebyshev's inequality and get that $\P(P_k(u,v) < \E P_k(u,v)/2) = o(1)$, so at least one such path exists with high probability.

Now we are ready to prove the theorem. Let us prove that $G(M)$-based NNS succeeds with probability $1 - o(1)$. 
It follows from Lemma~\ref{app::lem:varepsilon} (and the discussion below it) that under the conditions on $M$ and $L$, greedy $G(M)$-based NNS reaches the $L$-neighborhood of the query with probability $1 - o(1)$.

Thus, with probability $1-o(1)$ we reach the $L$-neighborhood within which there is a path to the nearest neighbor. 
%since the neighborhood is connected. 
Recall that we assume beam search with $\frac{C L^d}{\sqrt{d}}$ candidates. Choosing large enough $C$, we can guarantee that the number of candidates is larger than the number of elements in the $L$-neighborhood. This implies that all reachable elements inside the $L$-neighborhood will finally be covered by the algorithm.

%Let us note that the fact that there is a monotone path from $\x_L$ to $\bar \x$ does not guarantee that this path will be followed by the algorithm. The reason is that the algorithm is greedy with respect to $q$ and not $\bar \x$. Therefore, in our complexity analysis below we upper bound the number of local steps (i.e., steps inside the $L$-neighborhood) by the total size of the $L$-neighborhood.

Finally, it remains to analyze the time complexity. To reach the $L$-neighborhood, we need $\Theta\left(d^{1/2}\cdot \log n \cdot M^d \right)$ operations (recall that the number of steps can be bounded by $\log n$ due to long edges). Then, to fully explore the $L$-neighborhood, we need $O \left( L^d \cdot M^d \right)$. 
For $d > \log \log n$ the first term is negligible compared to the second one, so the required complexity follows.

\section{Comparison with the results of~\cite{laarhoven2018graph}}\label{sec:comp_w_laarh}

Here we extend the related work from the main text and discuss in more detail how our research differs from the results of~\cite{laarhoven2018graph}.

Laarhoven~\cite{laarhoven2018graph} analyzes time and space complexity for graph-based NNS in sparse regime when $d \gg \log n$. He considers plain NN graphs and allows multiple restarts. In contrast, we consider both regimes and assume only one iteration of a graph-based search. We do not consider multiple restarts since it is non-trivial to rigorously prove that restarts can be assumed ``almost independent'' (see Section~A.3, proof of Lemma 15, \cite{laarhoven2018graph}). 
As a result, for sparse datasets, we consider a slightly weaker setting with only one iteration, but all results are formally proven. Our result for sparse regime (Theorem~2 in the main text) corresponds to the case $\rho_q = \rho_s$ from~\cite{laarhoven2018graph}.

Also, in Section~\ref{app::sec:caps}, we state new bounds for the volumes of spherical caps' intersections, which are needed for the rigorous analysis in both sparse and dense regimes. We could not use the results of~\cite{becker2016new} since parameters defining spherical caps are assumed to be constant there, while they can tend to 0 or 1 in dense and sparse regimes.

We also address the problem of possible dependence between consecutive steps of the algorithm (Lemma~\ref{app::lem:dependence}). While we prove that it can be neglected, it is important for rigorous analysis. 

Most importantly, we analyze the dense regime and additional techniques (shortcuts and beam search), which are essential for the effective graph-based search. 
Interestingly, shortcut edges are useful only in the dense regime.

\section{Additional experiments}\label{sec:experiments_app}

\iffalse
\subsection{Datasets}\label{ref:datasets}

The properties of the datsets used in the current research are summarized in Table~\ref{tab:datasets}.

\begin{table}
\caption{Real datasets and their properties}
\label{tab:datasets}
\vspace{4pt}
\centering
 \begin{tabular}{||c c c c c||} 
 \hline
 dataset & dim  & No. of base & No. of query & metric \\ [0.5ex] 
 \hline\hline
 SIFT & 128  & $10^6$ & $10^4$ & Euclidean \\ 
 \hline
 GIST & 960  & $10^6$ & $10^3$ & Euclidean\\ 
 \hline
 GloVe & 300  & $10^6$ & $10^4$ & Angular \\
 \hline
 DEEP & 96 & $10^6$ & $10^4$ & Euclidean \\
 \hline
\end{tabular}
\end{table}
\fi

\subsection{Dense vs sparse setting}

Let us discuss our intuition on why real datasets are ``more similar'' to dense rather than sparse synthetic ones. 

In the sparse regime, all elements are almost at the same distance from each other, and even in the moderate regime ($d \propto \log(n)$), the distance to the nearest neighbor must be close to a certain constant. %that is independent of the dimension and the number of elements). 
In contrast, the dense regime implies high proximity of the nearest objects. 
While real datasets are always finite and asymptotic properties cannot be formally verified, we still can compare the properties of real and synthetic datasets. We plotted the distribution of the distance to the nearest neighbor (see Figure~\ref{fig:dist_disrt_1}) and see that for the SIFT dataset, the obtained distribution is more similar to the ones in the dense regime. This is further supported by the literature which estimates the intrinsic dimension of real data. For example, for the SIFT dataset with 128-dimensional vectors, the estimated intrinsic dimension is 16~\cite{levina2005lid}. Thus, we conclude that the analysis of the dense regime is important.

\begin{figure}
\centering
\includegraphics[width=0.6\textwidth]{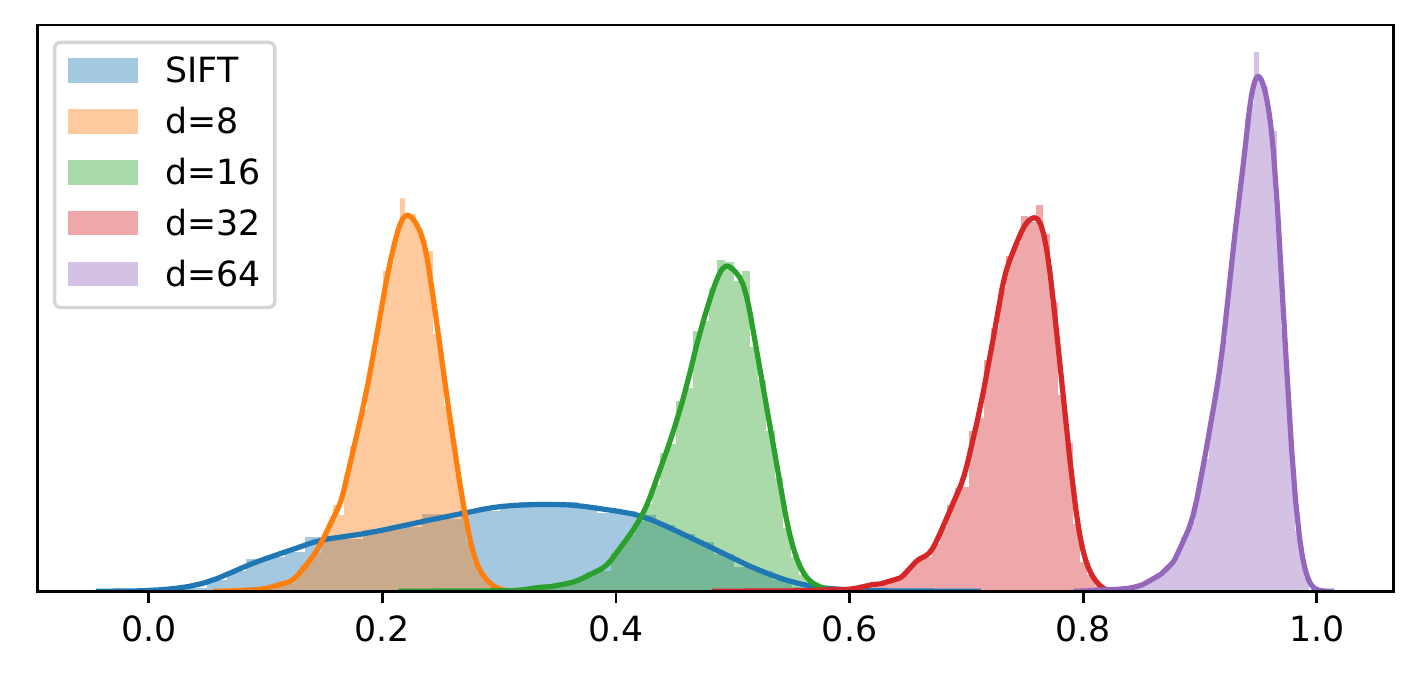}
\caption{The distribution of the distance to the nearest neighbor for the SIFT dataset and synthetic uniform data for different dimensions and the same size (1M)}
\label{fig:dist_disrt_1}
\end{figure}

\subsection{Parameters of algorithms}

In this section, we specify additional hyperparameters used in our experiments.

The number of edges used in \textsc{Kl}, when is not explicitly specified, is equal to 15, which is close to $\ln n$.

The number of edges in \textsc{kNN} graphs is dynamic when the beam search is not used. When the beam search is used, the number of edges for synthetic datasets is 8 for $d=2$, 10 for $d=4$, 16 for $d=8$, 20 for $d = 16$, and
 25 for all real datasets.

%beam  kl - we turn off kl-edges after 11, 7, 5, 4 steps respectively (and 7 for real datasets)
%hnsw=flat

The dimension we use for \textsc{dim-red} is 64 for GIST, 32 for SIFT, 48 for DEEP, 128 for GloVe. 

%gist: hnsw 18
%sift: hnsw 16
%deep: hnsw 16
%glove: hnsw 20, learned on efc=2000 (other 500)

\subsection{Additional experimental results}

In Figure~\ref{fig:optimal_appr}, we show that several approximations discussed in the main text do not affect the quality of graph-based NNS significantly (in the uniform case). Namely,
\begin{itemize}
    \item Connecting a node to other nodes at a distance smaller than some constant (\textsc{thrNN}) and to the fixed number of nearest neighbors (\textsc{kNN}) lead to graph-based algorithms with similar performance;
    \item Pre-sampling of $\sqrt{n}$ nodes when adding shortcut edges (\textsc{Sample}) lowers the quality, but not substantially;
    \item Rank-based probabilities for shortcut edges (\textsc{Kl-rank}) can lead to even better quality than distance-based (\textsc{Kl-dist}).
\end{itemize}

In Figure~\ref{fig:optimal_kl}, we illustrate how the number of long-range edges affects the quality of the algorithm. Let us note that $16$ is close to $\log n$ discussed in 
Corollary 2 of the main text.
%Corollary~\ref{cor:log_long}. 
Figure~\ref{fig:optimal_kl} shows that this value is indeed close to being optimal, especially for the high-accuracy regime, which is a focus of the current research. However, it also seems that the optimal number of long edges may depend on $d$: on Figure~\ref{fig:optimal_kl}, the relative performance of graphs with 32 long edges is improving as $d$ grows.

\begin{figure}
\centering
\includegraphics[width=\textwidth]{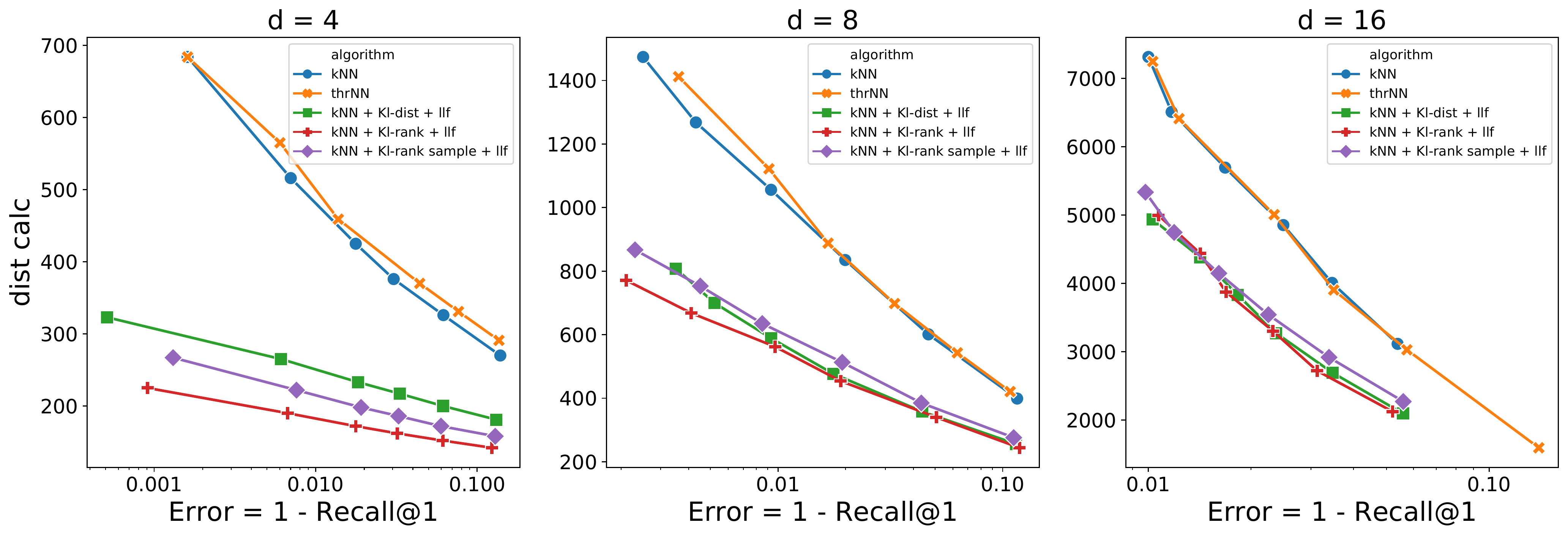}
\caption{The effect of \textsc{kNN} and \textsc{Kl} approximations} 
\label{fig:optimal_appr}
\end{figure}

\begin{figure}
\centering
\includegraphics[width=\textwidth]{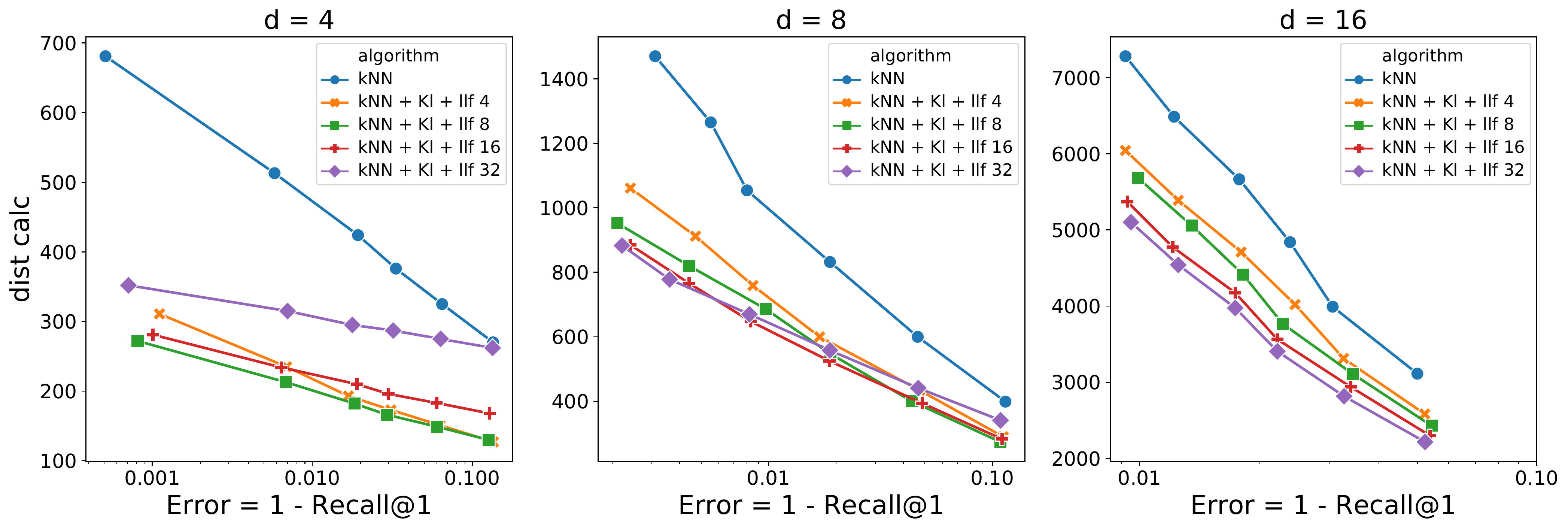}
\caption{The effect of the number of long-range edges}
\label{fig:optimal_kl}
\end{figure}

\bibliographystyle{abbrv}
\bibliography{neighbors}